\documentclass{article}

\usepackage[english]{babel}
\usepackage[utf8x]{inputenc}
\usepackage[T1]{fontenc}
\usepackage{slashed}

\usepackage[letterpaper,top=3.5cm,bottom=3.5cm,left=2.5cm,right=2.5cm,marginparwidth=2cm]{geometry}

\usepackage{enumitem}
\usepackage{amsmath}
\usepackage{mathtools}
\usepackage{graphicx}
\usepackage{amsfonts}
\usepackage{amsthm}
\usepackage{comment}
\usepackage[colorinlistoftodos]{todonotes}
\usepackage[colorlinks=true, allcolors=blue]{hyperref}
\usepackage{amssymb}
\usepackage{authblk}
\usetikzlibrary{decorations.markings}
\usetikzlibrary{shapes.geometric}
\newtheorem{lemma}{Lemma}[subsection]
\newtheorem{proposition}{Proposition}[section]
\newtheorem{theorem}{Theorem}[section]
\newtheorem{definition}{Definition}[section]
\newtheorem{remark}{Remark}[section]
\newtheorem{corollary}{Corollary}[section]
\newtheorem{conjecture}{Conjecture}
\newtheorem{assumption}{Assumption}

\title{Extensions of Lorentzian Hawking--Page Solutions with \ \ \ \ \ \ \ \ \ \  Null Singularities,
Spacelike Singularities, and Cauchy horizons of Taub--NUT type}
\author{Serban Cicortas\footnote{Princeton University, Department of Mathematics, Fine Hall, Washington Road, Princeton, NJ 08544, USA}}

\begin{document}

\maketitle

\begin{abstract}
    Starting from the Hawking--Page solutions of \cite{HP}, we consider the corresponding Lorentzian cone metrics. These represent cone interior scale-invariant vacuum solutions, defined in the chronological past of the scaling origin. We extend the Lorentzian Hawking--Page solutions to the cone exterior region in the class of $(4+1)$-dimensional scale-invariant vacuum solutions with an $SO(3)\times U(1)$ isometry, using the Kaluza--Klein reduction and the methods of Christodoulou in \cite{nakedsingularities}. We prove that each Lorentzian Hawking--Page solution has extensions with a null curvature singularity, extensions with a spacelike curvature singularity, and extensions with a null Cauchy horizon of Taub--NUT type. These are all the possible extensions within our symmetry class. The extensions to spacetimes with a null curvature singularity can be used to construct $(4+1)$-dimensional asymptotically flat vacuum spacetimes with locally naked singularities, where the null curvature singularity is not preceded by trapped surfaces. We prove the instability of such locally naked singularities using the blue-shift effect of Christodoulou in \cite{instability}.
\end{abstract}

\tableofcontents

\section{Introduction}

Given an $n$-dimensional Einstein manifold with negative scalar curvature, there is a canonical procedure to obtain the associated Lorentzian cone metric as an $(n+1)$-dimensional scale-invariant vacuum solution. This is defined in the chronological past of the scaling origin, so it can be viewed as a cone interior solution. For example, starting from the $n$-dimensional hyperbolic space, one constructs the Minkowski metric in the chronological past of the origin. In the case $n=3,$ all Einstein manifolds have constant sectional curvature, so the only $(3+1)$-dimensional Lorentzian cone metric is given by the above example. However, in the case $n=4$ there already exist non-trivial examples.

The problem of extending $(4+1)$-dimensional Lorentzian cone metrics to the cone exterior region as vacuum solutions was considered in a general context by Anderson in \cite{Andersoncone} and \cite{Anderson}. In these works, he provides heuristics suggesting that one could use this procedure in order to construct $(4+1)$-dimensional vacuum spacetimes with naked singularities. More precisely, he also defines the Lorentzian cone metric in the chronological future of the scaling origin, and extends the metric in the cone exterior region as a formal Wick rotation of the interior metric. If the scaling origin represents a curvature singularity and the exterior metric interpolates smoothly between the past and future cones, he argues that the scaling origin is a naked singularity. 

Constructing examples of naked singularities is of great importance in view of weak cosmic censorship, which is the conjecture that the maximal globally hyperbolic development of generic asymptotically flat initial data does not have naked singularities. The conjecture was proved for the model problem of Einstein-Scalar field in spherical symmetry by Christodoulou in \cite{instability}. Previously, Christodoulou had constructed examples of naked singularities for this matter model in \cite{nakedsingularities}, which implies that the genericity assumption in the statement of the conjecture is necessary. In vacuum, there has been remarkable recent progress in the works \cite{examplevacuumequations}, \cite{examplevacuumequationsinterior} and \cite{twistedselfsimilarity}, which constructed examples of $(3+1)$-dimensional vacuum spacetimes with naked singularities. However, the weak cosmic censorship conjecture remains one of the main open problems of general relativity. As an essential intermediate step in proving weak cosmic censorship, Christodoulou formulated the following conjecture in \cite{overview}, using the terminology of terminal indecomposable past set of \cite{Geroch}:

\begin{conjecture}\label{conjecture}
    For generic asymptotically flat initial data, the maximal future development $(M,g)$ has the property that if $\mathcal{P}$ is a terminal indecomposable past set in $M$ whose trace on the initial data hypersurface has compact closure $K$, then for any open domain $D$ containing $K$, the domain of dependence of $D$ in $M$ contains a closed trapped surface.
\end{conjecture}

\subsection{Hawking--Page Solutions and Main Results}

In this paper we provide a rigorous study of the problem considered by Anderson, in the case of the Lorentzian cone metrics corresponding to the Hawking--Page solutions of \cite{HP}, in order to investigate the possibility of constructing $(4+1)$-dimensional vacuum spacetimes with naked singularities. The Hawking--Page solutions are the 1-parameter family of Einstein metrics on $\mathbb{R}^2\times S^2:$
\[g^{(4)}_m=V^{-1}d\rho^2+\rho^2d\sigma_2^2+\bigg(\frac{2\rho_+}{3\rho_+^2+1}\bigg)^2Vd\sigma_1^2,\]
where for any $m>0$ we have $V(\rho)=\rho^2+1-\frac{2m}{\rho},$ and we denote its largest root by $\rho_+$. Moreover, $d\sigma_k^2$ represents the standard metric on $S^k.$ 

These solutions are of interest in the AdS/CFT correspondence physics literature, since according to \cite{Anderson} they provide examples of non-uniqueness to the problem of finding Einstein metrics with a given conformal type at infinity (see already Section \ref{non uniqueness}). 

We define the Lorentzian Hawking--Page solutions to be the associated cone metrics:
\[g^{(4+1)}_m=-d\tau^2+\tau^2V^{-1}d\rho^2+\tau^2\rho^2d\sigma_2^2+\bigg(\frac{2\rho_+}{3\rho_+^2+1}\bigg)^2\tau^2Vd\sigma_1^2.\]
They are scale-invariant $(4+1)$-dimensional vacuum spacetimes with an $SO(3)\times U(1)$ isometry and scaling vector field $S=\tau\partial_{\tau}$. Moreover, they represent cone interior solutions, in the sense that they are defined in the chronological past of the scaling origin. 

The main result of our paper is the construction of all the possible extensions of the Lorentzian Hawking--Page solutions to the cone exterior region in the class of $(4+1)$-dimensional scale-invariant vacuum solutions with an $SO(3)\times U(1)$ isometry. We obtain that each Lorentzian Hawking--Page solution has:
\begin{enumerate}[label=(\roman*)]
    \item extensions with a null curvature singularity (Theorem \ref{solutions singular null boundary}),
    \item extensions with a spacelike curvature singularity (Theorem \ref{solutions apparent horizon} and Theorem \ref{spacelike singularity}),
    \item extensions with a null Cauchy horizon of Taub--NUT type (Theorem \ref{solutions apparent horizon}, Theorem \ref{regular spacelike boundary}, Theorem \ref{regularity at BB}).
\end{enumerate}
We also describe the asymptotic behavior of the solutions near the boundaries. We point out that the solutions constructed are not asymptotically flat. In cases (i) and (ii), we represent the Penrose diagrams of the quotient of the spacetimes by the $SO(3)\times U(1)$ action:

\ \ 

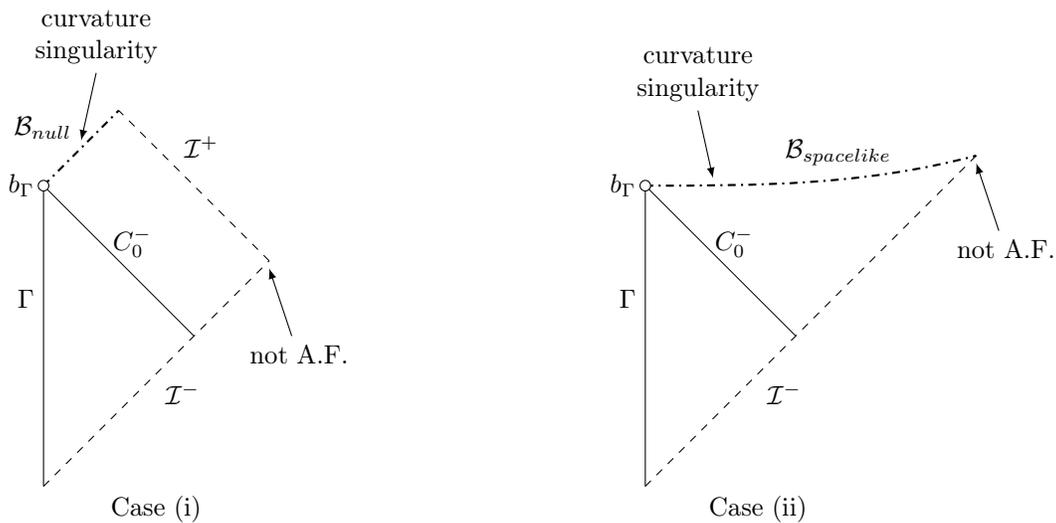
\begin{center}
\begin{tikzpicture}[scale=1]
\draw[dashed] (0,0) -- (3,3);
\draw (0,0) -- (0,4);
\draw[very thin] (0,4) -- (2,2);
\draw[dashed] (3,3) -- (1,5);
\draw[thick, dashdotted] (0,4) -- (1,5);
\filldraw[color=black, fill=white] (0,4) circle (2pt) node[anchor=east] {$b_{\Gamma}$};
\draw (0,2.5)  node[anchor=east] {$\Gamma$};
\draw (0.8,3.25)  node[anchor=west] {$C_0^-$};
\draw (1.5,1.5)  node[anchor=north west] {$\mathcal{I}^-$};
\draw (1.75,4.25)  node[anchor=south west] {$\mathcal{I}^+$};
\draw (0.5,4.5)  node[anchor=south east] {$\mathcal{B}_{null}$};
\draw (1.5,0)  node[anchor=north] {Case (i)};
\draw (0.7,6) node[anchor=south] {curvature};
\draw (0.7,5.5) node[anchor=south] {singularity};
\draw [-latex](0.7,5.5) -- (0.5,4.6);
\draw [-latex](3.3,2) -- (3,2.9);
\draw (3.4,2) node[anchor=north] {not A.F.};

\draw[dashed] (8,0) -- (12.4,4.4);
\draw (8,0) -- (8,4);
\draw[very thin] (8,4) -- (10,2);
\draw[thick, dashdotted] (8,4) .. controls (9.5,4) and (10.9,4) .. (12.4,4.4);
\filldraw[color=black, fill=white] (8,4) circle (2pt) node[anchor=east] {$b_{\Gamma}$};
\draw (8,2.5)  node[anchor=east] {$\Gamma$};
\draw (8.8,3.25)  node[anchor=west] {$C_0^-$};
\draw (9.5,1.5)  node[anchor=north west] {$\mathcal{I}^-$};
\draw (11.4,4.15)  node[anchor=south east] {$\mathcal{B}_{spacelike}$};
\draw (9.5,0)  node[anchor=north] {Case (ii)};
\draw (8.7,5.5) node[anchor=south] {curvature};
\draw (8.7,5) node[anchor=south] {singularity};
\draw [-latex](8.7,5) -- (8.9,4.1);
\draw [-latex](12.7,3.4) -- (12.4,4.3);
\draw (12.8,3.4) node[anchor=north] {not A.F.};
\end{tikzpicture}
\end{center}

In case (iii), we depict on the left the Penrose diagram of the quotient of the spacetime by the $SO(3)\times U(1)$ action. The boundary $\mathcal{CH}$ is null, because the $S^1$ direction is null. This is a Cauchy horizon of Taub--NUT type, since there exist non-unique analytic extensions beyond $\mathcal{CH}$ with closed timelike curves. The existence of these solutions is problematic from the point of view of strong cosmic censorship. We conjecture the Cauchy horizon of such solutions to be unstable. We represent the causality of one such extension in the diagram on the right:

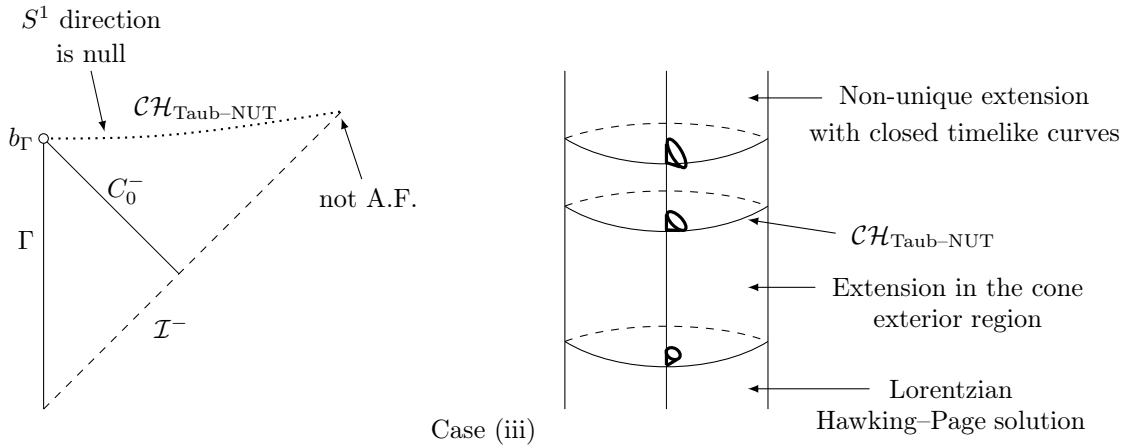
\begin{center}
\begin{tikzpicture}[scale=0.9]
\draw[dashed] (0.3,0) -- (4.7,4.4);
\draw (0.3,0) -- (0.3,4);
\draw[very thin] (0.3,4) -- (2.3,2);
\draw[thick, dotted] (0.3,4) .. controls (1.8,4) and (2.2,4) .. (4.7,4.4);
\filldraw[color=black, fill=white] (0.3,4) circle (2pt) node[anchor=east] {$b_{\Gamma}$};
\draw (0.3,2.5)  node[anchor=east] {$\Gamma$};
\draw (1.1,3.25)  node[anchor=west] {$C_0^-$};
\draw (1.8,1.5)  node[anchor=north west] {$\mathcal{I}^-$};
\draw (3.9,4.15)  node[anchor=south east] {$\mathcal{CH}_{\text{Taub--NUT}}$};
\draw (6.8,0)  node[anchor=north] {Case (iii)};
\draw (1,5.5) node[anchor=south] {$S^1$ direction};
\draw (1,5) node[anchor=south] {is null};
\draw [-latex](1,5) -- (1.2,4.1);
\draw [-latex](5,3.4) -- (4.7,4.3);
\draw (5.1,3.4) node[anchor=north] {not A.F.};

\draw (8,0) -- (8,5);
\draw (11,0) -- (11,5);
\draw[very thin] (9.5,0) -- (9.5,5);

\draw (13.3,2.55)  node {$\mathcal{CH}_{\text{Taub--NUT}}$};
\draw [-latex](12,2.6) -- (10.7,2.8);
\draw (8,3) .. controls (8.8,2.5) and (10.2,2.5) .. (11,3);
\draw[dashed] (8,3) .. controls (8.8,3.3) and (10.2,3.3) .. (11,3);
\draw[very thick, shift={(9.65,2.78)}, rotate=45] (0,0) ellipse (0.08 and 0.18);
\draw[very thick] (9.5,2.65) -- (9.5,2.88);
\draw[very thick] (9.48,2.65) -- (9.7,2.64);
    
\draw (13.9,4.6)  node {Non-unique extension};
\draw (13.9,4.1)  node {with closed timelike curves};
\draw [-latex](11.8,4.6) -- (10.7,4.6);
\draw (8,4) .. controls (8.8,3.5) and (10.2,3.5) .. (11,4);
\draw[dashed] (8,4) .. controls (8.8,4.3) and (10.2,4.3) .. (11,4);
\draw[very thick, shift={(9.65,3.78)}, rotate=30] (0,0) ellipse (0.08 and 0.24);
\draw[very thick] (9.5,3.65) -- (9.5,3.92);
\draw[very thick] (9.48,3.65) -- (9.7,3.58);

\draw (13.8,1.8)  node {Extension in the cone};
\draw (13.8,1.3)  node {exterior region};
\draw [-latex](11.7,1.8) -- (10.7,1.8);
    
\draw (13.7,0.3)  node {Lorentzian};
\draw (13.7,-0.2)  node { Hawking--Page solution};
\draw [-latex](12.5,0.3) -- (10.7,0.3);
\draw (8,1) .. controls (8.8,0.5) and (10.2,0.5) .. (11,1);
\draw[dashed] (8,1) .. controls (8.8,1.3) and (10.2,1.3) .. (11,1);
\draw[very thick, shift={(9.6,0.82)}, rotate=60] (0,0) ellipse (0.08 and 0.11);
\draw[very thick] (9.5,0.65) -- (9.5,0.88);
\draw[very thick] (9.5,0.64) -- (9.65,0.73);
\end{tikzpicture}
\end{center}

In particular, we obtain that in our situation the construction proposed by Anderson in \cite{Andersoncone} does not lead to naked singularities in the usual sense. We point out that a spacetime with a naked singularity would have the Penrose diagram of case (i) such that the boundary $\mathcal{B}_{null}$ is regular with complete null generators reaching infinity, and $\mathcal{I}^+$ is incomplete. However, in our situation we have that in case (i) the boundary $\mathcal{B}_{null}$ is a null curvature singularity, in case (ii) the boundary $\mathcal{B}_{spacelike}$ is a spacelike curvature singularity, and in case (iii) the generators of the Cauchy horizon $\mathcal{CH}_{\text{Taub--NUT}}$ are given by closed null curves. Moreover, we recall that the extension procedure to the cone exterior region suggested in \cite{Andersoncone} is to consider a Wick rotation of the cone interior solution. In the case of Lorentzian Hawking--Page solutions, we prove that these extensions correspond to case (iii) and we compute them explicitly.

We can use the spacetimes obtained in our main result in order to construct asymptotically flat extensions of Lorentzian Hawking--Page solutions following a standard gluing argument. In case (i) of spacetimes with a null curvature singularity, this construction gives $(4+1)$-dimensional asymptotically flat vacuum spacetimes with locally naked singularities, in the sense that the null curvature singularity is not preceded by trapped surfaces. We depict below the Penrose diagram of the quotient of such spacetimes by the $SO(3)\times U(1)$ action:
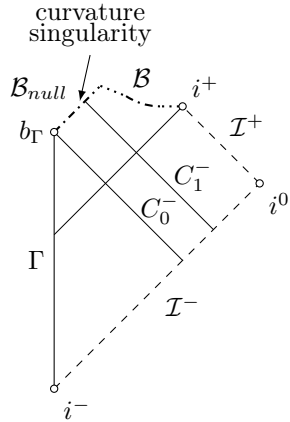
\begin{center}
\begin{tikzpicture}[scale=0.683]
\draw[dashed] (0,0) -- (4,4);
\draw (0,0) -- (0,5);
\draw[very thin] (0,3) -- (2.5,5.5);
\draw[very thin] (0,5) -- (2.5,2.5);
\draw[dashed] (4,4) -- (2.5,5.5);
\draw[thick, dashdotted] (0,5) -- (0.85,5.85);
\draw[thick, dotted] (0.85,5.85) -- (0.9,5.9);
\draw[very thin] (0.6,5.6) -- (3.1,3.1);
\draw[thick, dashdotted] (1.4,5.7) .. controls (1.7,5.5) and (2,5.5) .. (2.1,5.5);
\draw[thick, dotted] (2.1,5.5) -- (2.5,5.5);
\draw[thick, dotted] (0.9,5.9) -- (1.4,5.7);
\filldraw[color=black, fill=white] (0,5) circle (2pt) node[anchor=east] {$b_{\Gamma}$};
\filldraw[color=black, fill=white] (2.5,5.5) circle (2pt) node[anchor=south west] {$i^+$};
\filldraw[color=black, fill=white] (4,4) circle (2pt) node[anchor=north west] {$i^0$};
\filldraw[color=black, fill=white] (0,0) circle (2pt) node[anchor=north west] {$i^-$};
\draw (0,2.5)  node[anchor=east] {$\Gamma$};
\draw (1.5,3.5)  node[anchor=west] {$C_0^-$};
\draw (2.15,4.15)  node[anchor=west] {$C_1^-$};
\draw (2,2)  node[anchor=north west] {$\mathcal{I}^-$};
\draw (3.25,4.75)  node[anchor=south west] {$\mathcal{I}^+$};
\draw (0.45,5.45)  node[anchor=south east] {$\mathcal{B}_{null}$};
\draw (1.7,5.7)  node[anchor=south] {$\mathcal{B}$};
\draw (0.7,7) node[anchor=south] {curvature};
\draw (0.7,6.5) node[anchor=south] {singularity};
\draw [-latex](0.7,6.5) -- (0.5,5.6);
\end{tikzpicture}    
\end{center}

We prove the instability of such locally naked singularities using the blue-shift effect of Christodoulou in \cite{instability}, and interpret this result in the context of Conjecture \ref{conjecture}.

\subsection{Idea of the Proof}

A key aspect is the Kaluza--Klein reduction, which provides a connection between our work and the construction of examples of naked singularities of Christodoulou in \cite{nakedsingularities}. In general, the Kaluza--Klein reduction gives a correspondence between $(4+1)$-dimensional vacuum metrics with a $U(1)$ isometry and $(3+1)$-dimensional solutions to the Einstein-Scalar field equations. The additional symmetries of the Lorentzian Hawking--Page solutions imply that the associated Einstein-Scalar field solutions are spherically symmetric and self-similar, which is exactly the framework considered in \cite{nakedsingularities}. We point out that in our case the Einstein-Scalar field solutions have exotic boundary conditions at the center, dictated by regularity of the corresponding $(4+1)$-dimensional vacuum spacetimes.

We obtain that the $(4+1)$-dimensional Einstein vacuum equations in our setting reduce to an autonomous system of first order ODEs in the plane. The Lorentzian Hawking--Page solutions can be identified with orbits of this system. The future Cauchy horizons of the Lorentzian Hawking--Page solutions correspond to singular points of the orbits. A careful analysis allows us to continue the orbits beyond the singular points, which amounts to extending the solutions to the cone exterior region with a loss of regularity across the cone. In particular, the solutions are $C^{1,1/2}$ across the cone and smooth everywhere else. We remark that this mild singularity also occurs for the examples constructed in \cite{nakedsingularities}, and we argue that our solutions belong to a well-posedness class. The ODE analysis proves that each interior Lorentzian Hawking--Page solution bifurcates to a one-parameter family of exterior solutions.

We firstly consider case (i) of extensions to solutions with a null singular boundary, where the curvature blows up. We later use these extensions to construct $(4+1)$-dimensional asymptotically flat vacuum spacetimes with locally naked singularities. We finally prove their instability using the blue-shift effect of \cite{instability}.

The extensions of the Lorentzian Hawking--Page solutions in the cone exterior region in cases (ii) and (iii) have a spacelike boundary $\mathcal{B}$ in the Penrose diagrams of the quotients by the $SO(3)\times U(1)$ action. We compute the precise asymptotic behavior of the solutions near $\mathcal{B}$, and remark that the geometric interpretation of this boundary in the upstairs spacetimes is different in the two cases. Thus, in case (ii) we have that $\mathcal{B}=\mathcal{B}_{spacelike}$ represents a spacelike curvature singularity. However, in case (iii) we prove that $\mathcal{B}=\mathcal{CH}_{\text{Taub--NUT}}$ is a null Cauchy horizon where the $S^1$ direction is null. We prove that there exist non-unique analytic extensions beyond $\mathcal{CH}_{\text{Taub--NUT}}$ with closed timelike curves, which is a similar behavior with the Taub--NUT spacetimes. An essential step in dealing with case (iii) is that we can compute these extensions in the cone exterior region explicitly, which can be interpreted as a Wick rotation of the interior Lorentzian Hawking--Page solutions, as in \cite{Anderson} (see already Section \ref{regular spacelike boundary section}). 

We point out that in both cases (ii) and (iii), the corresponding Einstein-Scalar field solutions have a spacelike curvature singularity at $\mathcal{B}.$ We compute the blow-up rate of the scalar field $\phi=c\log r+O(1)$ near the spacelike singularity, where $c=\pm\sqrt{3},\ \mp\frac{1}{\sqrt{3}}$. The analysis of the orbits of the above solutions is also of interest because it describes the behavior in the trapped region of the Einstein-Scalar field solutions of \cite{nakedsingularities} in the case $k^2=\frac{1}{3}$.

\subsection{Outline of the Paper}

We outline the remainder of the paper. In Section \ref{first big section}, we study the Lorentzian Hawking--Page solutions in the cone interior region in detail. We introduce the Kaluza--Klein reduction and we prove that in our symmetry class the $(4+1)$-dimensional Einstein vacuum equations reduce to an autonomous system of first order ODEs in the plane. We identify the Lorentzian Hawking–Page solutions with orbits of this system. In Section \ref{second big section}, we study all the possible extensions of our solutions to the cone exterior region in the class of $(4+1)$-dimensional scale-invariant vacuum solutions with an $SO(3)\times U(1)$ isometry. In the non-trapped region, we describe the extensions with a null curvature singularity in Theorem \ref{solutions singular null boundary} and the extensions with an apparent horizon in Theorem \ref{solutions apparent horizon}. Continuing the latter in the trapped region, we describe the extensions with a spacelike curvature singularity in Theorem \ref{spacelike singularity} and the extensions to the explicit solutions with a null Cauchy horizon of Taub--NUT type in Theorem \ref{regular spacelike boundary} and Theorem \ref{regularity at BB}. In Section \ref{third big section}, we use the extensions from Theorem \ref{solutions singular null boundary} in order to construct asymptotically flat spacetimes with locally naked singularities. We then prove their instability in the $SO(3)\times U(1)$ symmetry class.

\subsection{Acknowledgements}

The author would like to greatly acknowledge Igor Rodnianski for his valuable help in the process of writing this paper. The author would also like to thank Mihalis Dafermos for the very useful discussions.

\section{Lorentzian Hawking--Page Solutions in the Cone Interior Region}\label{first big section}

We begin the section by proving some basic properties of the Hawking--Page solutions and the associated Lorentzian Hawking--Page solutions. We then introduce the Kaluza--Klein reduction, which gives a connection to self-similar spherically symmetric solutions of the Einstein-Scalar field equations. This allows us to work in the framework of Christodoulou in \cite{nakedsingularities}, so we obtain that the Einstein vacuum equations for $(4+1)$-dimensional scale-invariant spacetimes with an $SO(3)\times U(1)$ isometry reduce to an autonomous system of first order ODEs in the plane. We provide a phase portrait analysis of our autonomous system. Finally, we identify the Lorentzian Hawking--Page solutions with orbits of this system and prove that they represent cone interior solutions.

\subsection{Hawking--Page Solutions}\label{Hawking--Page solutions section}

The Hawking--Page solutions are a 1-parameter family of Einstein metrics on $\mathbb{R}^2\times S^2$, first introduced in \cite{HP}. Following the presentation in \cite{Anderson}, for any $m>0$ we consider the function:
\[V(\rho)=\rho^2+1-\frac{2m}{\rho}.\]
Let $\rho_+$ be the largest root of $V$ and define:
\[b=\frac{4\pi \rho_+}{3\rho_+^2+1}.\]

\begin{definition}
    We define the Hawking--Page metrics on $(\rho_+,\infty)\times S^2\times S^1$ to be:
\[g^{(4)}_m=V^{-1}d\rho^2+\rho^2d\sigma_2^2+\bigg(\frac{b}{2\pi}\bigg)^2Vd\sigma_1^2.\]
\end{definition}

\begin{proposition}
    The Hawking--Page solutions are Riemannian metrics on $\mathbb{R}^2\times S^2$ with an $SO(3)\times U(1)$ isometry. They are Einstein metrics satisfying $Ric(g)=-3g.$ The $SO(3)$ action is free, and the set of fixed points of the $U(1)$ action represents a regular center given by the round 2-sphere of radius $\rho_+.$
\end{proposition}
\begin{proof}
It suffices to change from polar to cartesian coordinates on $\mathbb{R}^2$. We introduce the radius function:
\[R(\rho)=\int_{\rho_+}^\rho\frac{d\tilde{\rho}}{\sqrt{V(\tilde{\rho})}}.\]
The desired change of coordinates is given by $(\rho,\theta,\varphi,\psi)\xrightarrow[]{}(x,y,\theta,\varphi)$, where $x=R(\rho)\cos\psi,\ y=R(\rho)\sin\psi.$ The rest of the proof is a simple computation.
\end{proof}

\subsection{Lorentzian Cone Metrics}

Following the definitions in \cite{BV} and \cite{nakedsingularities}, we say that a metric is scale-invariant if there exists a 1-parameter group of diffeomorphisms $h_a$ such that $h_a^*g=a^2g.$ We refer to the vector field $S$ generating the group $\{h_a\}$ as the scaling vector field. 

We note here the following canonical procedure of \cite{Andersoncone} to construct Lorentzian cone metrics:

\begin{proposition}
Let $(M,g^{(4)})$ be a 4-dimensional Riemannian manifold. Then the $(4+1)$-dimensional Lorentzian manifold $\big((0,\infty)\times M,g^{(4+1)}\big),$ where:
\[g^{(4+1)}=-d\tau^2+\tau^2g^{(4)},\]
is a vacuum solution if and only if $g^{(4)}$ is an Einstein metric with $Ric(g^{(4)})=-3g^{(4)}$. Moreover, the metric $g^{(4+1)}$ is scale-invariant with scaling vector field $S=\tau\partial_{\tau}$.
\end{proposition}
\begin{proof}
    We write down the nonzero Christoffel symbols and the components of the Ricci curvature:
    \[{}^{(4+1)}\Gamma_{ij}^k={}^{(4)}\Gamma_{ij}^k,\ {}^{(4+1)}\Gamma_{ij}^{\tau}=\tau g_{ij},\ {}^{(4+1)}\Gamma_{{\tau}i}^i=\frac{1}{{\tau}},\]
    \[{}^{(4+1)}R_{ij}={}^{(4)}R_{ij}+3g_{ij},\ {}^{(4+1)}R_{{\tau}i}={}^{(4+1)}R_{{\tau}{\tau}}=0.\]
    To prove the scale invariance statement, we check directly that $\mathcal{L}_Sg^{(4+1)}=2g^{(4+1)}.$
\end{proof}

As a corollary of this proposition we extend the Hawking--Page solutions to Lorentzian metrics:

\begin{definition}
    We define the Lorentzian Hawking--Page metrics on $\mathcal{M}_{int}^{(4+1)}=(0,\infty)\times\mathbb{R}^2\times S^2$ to be:
    \[g_m^{(4+1)}=-d\tau^2+\tau^2g_m^{(4)}.\]
\end{definition}

We remark that the Lorentzian Hawking--Page metrics are scale-invariant vacuum solutions and have an $SO(3)\times U(1)$ isometry. As before, the set of fixed points of the $U(1)$ action represents a regular center given by $\Gamma^{(4+1)}=(0,\infty)\times S^2,$ with $\{\rho=\rho_+\}.$ Using polar coordinates, we denote the non-central component by $\mathcal{U}^{(4+1)}=(0,\infty)\times(\rho_+,\infty)\times S^2\times S^1$. On $\mathcal{U}^{(4+1)},$ the Lorentzian Hawking--Page metrics have the form:
\[g^{(4+1)}_m=-d\tau^2+\tau^2V^{-1}d\rho^2+\tau^2\rho^2d\sigma_2^2+\bigg(\frac{b}{2\pi}\bigg)^2\tau^2Vd\sigma_1^2.\]

\subsection{Kaluza--Klein Reduction}\label{Kaluza Klein}

There is a canonical procedure to reduce a $(4+1)$-dimensional vacuum metric with a $U(1)$ isometry to a $(3+1)$-dimensional solution to the Einstein-Scalar field equations. The latter is a well studied system in spherical symmetry, see the works of Christodoulou \cite{formationoftrappedsurfaces}--\cite{overview}, and also \cite{trappedsurface},\cite{extensionprinciple}, and \cite{kommemi}. We shall use extensively these results in order to study the behavior of the $(4+1)$-dimensional vacuum solutions.

\begin{proposition}\label{Kaluza Klein proposition}
Let $(\Tilde{M},\Tilde{g}^{(4+1)})$ be a $(4+1)$-dimensional spacetime with a $U(1)$ isometry, such that $\Tilde{M}=M\times S^1$ and $\Tilde{g}^{(4+1)}=\Tilde{g}^{(3+1)}+f^2d\sigma_1^2.$ Then $(\Tilde{M},\Tilde{g}^{(4+1)})$ is a vacuum spacetime if and only if $(M,g,\phi)$ solves the Einstein-Scalar field equations, where $g=f\Tilde{g}^{(3+1)},$
\[f=\exp\bigg(\pm\frac{2}{\sqrt{3}}\phi\bigg).\]
\end{proposition}
\begin{proof}
We compute the nonzero Christoffel symbols:
    \[\Tilde{\Gamma}_{ij}^k=\Gamma_{ij}^k-\frac{1}{2}\big(\delta_j^k\partial_i\log f+\delta_i^k\partial_j\log f-g_{ij}\partial^k\log f\big),\]
    \[\Tilde{\Gamma}_{\psi\psi}^i=-f^3\partial^i\log f,\ \Tilde{\Gamma}_{\psi i}^{\psi}=\partial_i\log f.\]
We then compute the components of the Ricci curvature:
\[\tilde{R}_{ij}=R_{ij}-\frac{3}{2}\partial_i\log f\partial_j\log f+\frac{1}{2}g_{ij}\square_g\log f,\]
\[\tilde{R}_{i\psi}=0,\]
\[\tilde{R}_{\psi\psi}=-f^3\square_g\log f.\]
Therefore, the Einstein vacuum equations for $(\Tilde{M},\Tilde{g}^{(4+1)})$ are equivalent to the Einstein-Scalar field equations for $(M,g,\phi),$ namely:
\[\begin{cases}
R_{ij}=2\partial_i\phi\partial_j\phi, \\
\square_g\phi=0.
\end{cases}\]
\end{proof}

We can now apply this result to the non-central component of the Lorentzian Hawking--Page solutions:

\begin{definition}
    We define the Hawking--Page Einstein-Scalar field solutions on $\mathcal{U}^{(3+1)}=(0,\infty)\times(\rho_+,\infty)\times S^2:$
    \[g_m=-\frac{b}{2\pi}\tau\sqrt{V}d\tau^2+\frac{b}{2\pi}\frac{\tau^3}{\sqrt{V}}d\rho^2+r^2d\sigma_2^2,\]
\[r=\sqrt{\frac{b}{2\pi}}\tau^{3/2}V^{1/4}\rho,\]
\[\phi_m=\pm\frac{\sqrt{3}}{2}\log\bigg(\frac{b}{2\pi}\tau\sqrt{V}\bigg).\]
\end{definition}

We remark that the above solutions are spherically symmetric. The set of fixed points of the $SO(3)$ action is given by $\Gamma^{(3+1)}=(0,\infty)\times\{\rho=\rho_+\}\times S^2$ and it represents the singular center. 

Our solutions have an additional important symmetry, satisfying the self-similarity condition introduced in \cite{nakedsingularities}. We recall that a solution of the Einstein-Scalar field equations is self-similar if there exists a 1-parameter group of diffeomorphisms $h_a$ such that: \[h_a^*g=a^2g,\ h_a^*\phi=\phi-k\log a.\]
Equivalently, denoting by $S$ the scaling vector field generating $\{h_a\},$ the solution is self-similar if:
\[\mathcal{L}_Sg=2g,\ S\phi=-k.\]
We remark that if $k=0$, we recover the notion of scale invariance. The more general notion of self-similarity combines the invariance of the Einstein-Scalar field equations under rescaling the metric and translating the scalar field. By direct computation we obtain:

\begin{proposition}
    The Hawking--Page Einstein-Scalar field solutions are self-similar solutions with $S=\frac{2}{3}\tau\partial_{\tau}$ and $k=\mp\frac{1}{\sqrt{3}}$.
\end{proposition}

In the general case, the same computation shows that a $(4+1)$-dimensional vacuum spacetime $(\Tilde{M},\Tilde{g}^{(4+1)})$ with a $U(1)$ isometry is scale-invariant, with scaling vector field $\tilde{S}$, if and only if the $(3+1)$-dimensional Einstein-Scalar field solution $(M,g,\phi)$ obtained via the Kaluza--Klein reduction is self-similar, with scaling vector field $S=\frac{2}{3}\tilde{S}$ and $k=\mp\frac{1}{\sqrt{3}}$.

According to this result, we expect to have a connection with the analysis of Christodoulou in \cite{nakedsingularities}, where he studies self-similar spherically symmetric solutions of the Einstein-Scalar field equations. We shall explore this connection in the following section.

We point out that the Hawking--Page Einstein-Scalar field solutions are not globally hyperbolic, since they have a singular center. However, they are the unique evolution of the initial data in the self-similar class. Moreover, the corresponding vacuum Lorentzian Hawking--Page solutions are globally hyperbolic and represent the unique development of the smooth initial data induced on $\mathbb{R}^2\times S^2.$

\begin{remark}
    We could repeat this computation in the general case of $\Tilde{M}=M\times N$, with $\dim M=m, \dim N=n,$ and the product metric:
    \[\tilde{g}_{\alpha\beta}dx^{\alpha}dx^{\beta}=\frac{1}{f(x^1,\cdots,x^m)}g^M_{ij}dx^idx^j+f^A(x^1,\cdots,x^m)g^N_{ab}dx^adx^b.\]
    We obtain:
    \[\Tilde{R}_{ij}=R^M_{ij}-\frac{A^2n+2An-m+2}{4}\partial_i\log f\partial_j\log f+\frac{1}{2}g^M_{ij}\square_{g^M}\log f-\]\[-\big(\frac{An}{2}+1-\frac{m}{2}\big)\bigg[\nabla^M_i\nabla^M_j
    \log f-\frac{1}{2}\partial_k\log f\partial^k\log fg^M_{ij}\bigg],\]
    \[\Tilde{R}_{ai}=0,\]
    \[\Tilde{R}_{ab}=R^N_{ab}-\frac{A}{2}f^{A+1}\bigg[\square_{g^M}\log f+\big(\frac{An}{2}+1-\frac{m}{2}\big)\partial_i\log f\partial^i\log f\bigg]g^N_{ab}.\]
    To obtain a correspondence between a vacuum metric on $\tilde{M}$ and an Einstein-Scalar field solution on $M,$ we require that $m=4$ and $An=2.$ We also require $M$ to be spherically symmetric. If we attach a regular center $\Gamma$ to $\tilde{M}$, corresponding to $f=0,$ then the Cauchy hypersurface $\Sigma$ of $\tilde{M}$ has the topology of $CN\times S^2,$ where $CN$ is the cone $\big(N\times[0,\infty)\big)/(N\times\{0\}).$ In order for $\Sigma$ to be a manifold at the center, we require $N=S^n.$ However, $S^2,\ S^3$ do not admit Ricci flat metrics, whereas this problem is still open in the case $n>3.$ Therefore, we must take $N=S^1$ and $A=2,$ as in Proposition \ref{Kaluza Klein proposition}.
\end{remark}

\subsection{Scale-invariant Vacuum Spacetimes}\label{general self similar solutions}

Motivated by the Lorentzian Hawking--Page solutions, we study general scale-invariant vacuum spacetimes with an $SO(3)\times U(1)$ isometry. We prove that in this case, the Einstein vacuum equations reduce to a two dimensional autonomous system of first order
ODEs, and we describe the phase portrait of this system.

We make the following assumption on the spacetime:

\begin{assumption}\label{assumption}
    The spacetime is globally hyperbolic with Cauchy hypersurface $\Sigma=\mathbb{R}^2\times S^2$. The $SO(3)$ action is free. The set of fixed points of the $U(1)$ action represents a regular timelike center $\Gamma$ with topology $\mathbb{R}\times S^2.$
\end{assumption}

We denote the non-central component of our spacetime by $\mathcal{U}^{(4+1)}.$ We also denote by $\mathcal{U}^{(1+1)}$ the two dimensional Lorentzian quotient manifold. The general form of such a metric on $\mathcal{U}^{(1+1)}\times S^2\times S^1$ is:
\[\Tilde{g}=\frac{1}{f}\cdot g^{(1+1)}+\frac{1}{f}r^2d\sigma_2^2+f^2d\sigma_1^2,\]
where we denoted by $\frac{1}{f}g^{(1+1)}$ the metric induced on $\mathcal{U}^{(1+1)}$. According to our assumption, we have that: \[f=0,\ \frac{r^2}{f}\neq0\text{ on }\Gamma.\]

Using Proposition \ref{Kaluza Klein proposition}, we get that $\big(\mathcal{U}^{(1+1)}\times S^2,g,\phi\big)$ solves the Einstein-Scalar field equations, where:
\[g=g^{(1+1)}+r^2d\sigma_2^2,\ \phi=\pm\frac{\sqrt{3}}{2}\log f.\]
We obtain that this solution is self-similar with $k=\mp\frac{1}{\sqrt{3}}$, as a consequence of the scale invariance of the $(4+1)$-dimensional vacuum metric. Moreover, the center given in the limit as $r\xrightarrow[]{}0$ is singular. We abuse notation and denote it once again by $\Gamma$ (or $\Gamma^{(3+1)}$), since it is always clear from the context whether we refer to the $(3+1)$-dimensional or the $(4+1)$-dimensional solution.

As we pointed out before, such solutions of the  Einstein-Scalar field equations were also considered by Christodoulou in \cite{nakedsingularities}. The difference is that Christodoulou's solutions are regular $(3+1)$-dimensional solutions, whereas the solutions that we consider are singular at the center. However, in the latter case the behavior of the center is dictated by regularity of the corresponding $(4+1)$-dimensional scale-invariant vacuum spacetimes.

In the following, we shall carry out most of our analysis in the context of self-similar spherically symmetric solutions to the Einstein-Scalar field equations, by adapting Christodoulou's method to the exotic boundary conditions at the center. We then use these results to obtain information about the associated $(4+1)$-dimensional vacuum metrics.
\begin{remark}\label{conical singularity remark}
    Solutions to the Einstein-Scalar field equations have the gauge freedom $\phi\xrightarrow[]{}\phi+C$. However, the corresponding vacuum spacetimes obtained through this transformation fail to satisfy the regularity at the center in Assumption \ref{assumption}. This can already be seen for Lorentzian Hawking--Page solutions, which have a conical singularity at the center with respect to the $\tau$ foliation if we take $C\neq0:$
    \[\lim_{\rho\xrightarrow[]{}\rho_+}\frac{\int_{\rho_+}^{\rho}\frac{\tau}{\sqrt{C'V}}}{\frac{b}{2\pi}\tau C'\sqrt{V}}=\frac{1}{C'\sqrt{C'}}\neq1,\]
    where $C'=\exp\big(\pm\frac{2C}{\sqrt{3}}\big)$ and we used the computation in Section \ref{Hawking--Page solutions section}. Moreover, we prove that if we take $C\neq0$, the solutions have a conical singularity at the center with respect to any foliation by regular spacelike hypersurfaces. We consider any regular change of coordinates $\tau(T,R),\ \rho(T,R)$ such that the $T$ slices are spacelike hypersurfaces and we denote the center by $\Gamma=\{(T,R):\ R=R_+(T)\}\times S^2.$ For the coordinates to be smooth at the center, we require that on slices of constant $T$ we have:
    \[\lim_{R\rightarrow R_+(T)}\frac{\partial\tau}{\partial R}=0,\ \lim_{R\rightarrow R_+(T)}\frac{\partial\rho}{\partial R}\neq0.\]
    We repeat the above computation for the metric induced on slices of constant $T,$ and obtain that the solutions have a conical singularity at the center with respect to the $T$ foliation if we take $C\neq0.$
\end{remark}

\subsubsection{Self-similar Bondi gauge}\label{introduction to bondi gauge}

We begin by making the following assumption on the vacuum spacetime:

\begin{assumption}\label{assumption non trapped}
    On the initial data hypersurface $\Sigma,$ there exists a neighborhood of $\Gamma$ that is non-trapped.
\end{assumption}

We compute that the area of an orbit $S^2\times S^1$ of the $SO(3)\times U(1)$ action is $8\pi^2r^2.$ Similarly, the area of the corresponding orbit $S^2$ of the $SO(3)$ action on the $(3+1)$-dimensional spacetime obtained via the Kaluza--Klein reduction is $4\pi r^2.$ Thus, there is a correspondence between trapped codimension 2 surfaces in $(\mathcal{U}^{(4+1)},\Tilde{g})$ and trapped codimension 2 surfaces in $(\mathcal{U}^{(1+1)}\times S^2,g).$ The above assumption implies that in the $(3+1)$-dimensional spacetime there exists a non-trapped neighborhood of the center on the initial data hypersurface $\big(\Sigma\cap\mathcal{U}^{(4+1)}\big)/U(1).$ As a result, we can introduce Bondi coordinates on $\mathcal{U}^{(1+1)}\times S^2$, which allow us to work precisely in the framework of \cite{nakedsingularities}. We point out that we prefer Bondi coordinates to double null coordinates for now because they are more suitable for constructing the solution from the center.

In Bondi coordinates on $\mathcal{U}^{(1+1)}\times S^2$ we have:
\begin{equation}\label{bondi 1}
    g=-e^{2\nu}du^2-2e^{\nu+\lambda}dudr+r^2d\sigma_2^2.
\end{equation} 
In these coordinates, the associated $(4+1)$-dimensional metric on $\mathcal{U}^{(1+1)}\times S^2\times S^1$ has the form:
\[\Tilde{g}=-\frac{1}{f}e^{2\nu}du^2-\frac{2}{f}e^{\nu+\lambda}dudr+\frac{1}{f}r^2d\sigma_2^2+f^2d\sigma_1^2.\]

We also make an assumption on the scale invariance of $\big(\mathcal{U}^{(4+1)},\tilde{g}\big):$

\begin{assumption}\label{assumption scale}
    The spacetime $\big(\mathcal{U}^{(4+1)},\tilde{g}\big)$ is scale-invariant, having the one parameter group of diffeomorphism $\tilde{h}_a$ and scaling vector field $\tilde{S}=\frac{d}{da}\big|_{a=1}\tilde{h}_a.$ The spacetime $\big(\mathcal{U}^{(1+1)}\times S^2,g,\phi\big)$ is self-similar, having the one parameter group of diffeomorphism $h_a$ and scaling vector field $S=\frac{d}{da}\big|_{a=1}h_a=\frac{2}{3}\tilde{S}.$ We assume that $h_a(u,r)=\big(h^1_a(u,r),h^2_a(u,r)\big),$ with $\frac{d}{da}\big|_{a=1}h^1_a>0,\ \frac{d}{da}\big|_{a=1}h^2_a=r.$
\end{assumption}

The self-similar condition for $\big(\mathcal{U}^{(1+1)}\times S^2,g,\phi\big)$ implies that:
\begin{equation}\label{bondi 2}
    Sr=r,\ S\phi=-k,\ \mathcal{L}_Sg_{\mathcal{U}^{(1+1)}}=2g_{\mathcal{U}^{(1+1)}}.
\end{equation}
We assumed that $S=\big(\frac{d}{da}\big|_{a=1}h^1_a\big)\frac{\partial}{\partial u}+r\frac{\partial}{\partial r},$ which is consistent with $Sr=r.$ Also, we notice that the condition $\mathcal{L}_Sg_{rr}=0$ implies that $\frac{d}{da}\big|_{a=1}h^1_a$ is independent of $r,$ so we can make the change of coordinates $u=-\exp{\int1/\big(\frac{d}{da}\big|_{a=1}h^1_a\big)}.$ As a consequence of our assumption we obtain Bondi coordinates with $u<0$ and: 
\begin{equation}\label{bondi 3}
    S=u\partial_u+r\partial_r.
\end{equation}
Following \cite{nakedsingularities}, we say that $\big(\mathcal{U}^{(1+1)}\times S^2,g,\phi\big)$ is in \textit{self-similar Bondi gauge} if there exist such coordinates $(u,r)$ on $\mathcal{U}^{(1+1)}$. The gauge freedom left is:
\begin{equation}\label{gauge freedom}
    u\xrightarrow[]{}c^2u.
\end{equation}
In this gauge, the self-similar condition (\ref{bondi 2}) becomes:
\begin{equation}
    S\nu=0,\ S\lambda=0,\ S\phi=-k.
\end{equation}

We briefly recall here the standard definitions that one makes in Bondi coordinates, as in \cite{nakedsingularities}, \cite{instability}. We define the future null outgoing  and incoming vector fields:
\[l=e^{-\lambda}\partial_r,\ n=2e^{-\nu}\partial_u-e^{-\lambda}\partial_r.\]
We then define the renormalized derivatives of $\phi$ in the outgoing and incoming directions:
\[\theta=r\frac{l\phi}{lr},\ \zeta=r\frac{n\phi}{nr}.\]
The Hawking mass satisfies the relation:
\[1-\frac{2m}{r}=e^{-2\lambda}.\]

We now introduce coordinates adapted to self-similar solutions. We consider:
\[e^s=x=-\frac{r}{u}.\]
In $(u,x)$ coordinates we have $S=u\partial_u.$ Finally, we define:
\[\beta=\frac{1}{\alpha}=1-
\frac{e^{\nu-\lambda}}{2e^s}.\]
According to \cite{nakedsingularities}, the Einstein-Scalar field system reduces to the system:
\begin{equation}\label{system}
    \begin{dcases}
    \frac{d\alpha}{ds}=\alpha\big[(\theta+k)^2+(1-k^2)(1-\alpha)\big] \\
    \frac{d\theta}{ds}=k\alpha\big(k\theta-1\big)+\theta\big[(\theta+k)^2-(1+k^2)\big]
    \end{dcases}
\end{equation}
Moreover, we have:
\begin{equation}\label{zeta}
    \zeta=-\frac{\theta+k\alpha}{\alpha-1},
\end{equation}
\begin{equation}\label{e{2lambda}}
    e^{2\lambda}=1+k^2+\frac{(\theta+k)^2}{\alpha-1}.
\end{equation}

We use this system to obtain more information about general vacuum spacetimes satisfying Assumptions \ref{assumption}, \ref{assumption non trapped} and \ref{assumption scale}, in particular about Lorentzian Hawking--Page spacetimes. We notice we can also start from a solution of (\ref{system}) and construct vacuum spacetimes by imposing the necessary regularity conditions. This latter approach is used in \cite{nakedsingularities}, and we shall also use it in order to extend the  Lorentzian Hawking--Page spacetimes.

\subsubsection{Boundary condition at the center}\label{Boundary condition at the center}

For each orbit of the autonomous system (\ref{system}) such that the expressions corresponding to $e^{2\lambda}$ and $e^{\nu-\lambda}$ are positive, we obtain a solution to the Einstein-Scalar field system satisfying Assumptions \ref{assumption non trapped} and \ref{assumption scale}. We claim that Assumption \ref{assumption} on regularity at the center of the $(4+1)$-dimensional vacuum solution determines the asymptotic behavior of the orbits of (\ref{system}), and can be viewed as a boundary condition at the center.

Indeed, we assumed that $r=0$ on $\Gamma,$ so in $(u,s)$ coordinates the center is given by $s\xrightarrow[]{}-\infty.$ We also assumed that $r^2/f\neq 0$ on $\Gamma,$ so $f\sim r^2$ as $s\xrightarrow[]{}-\infty,$ which is equivalent to  $\phi\mp\sqrt{3}\log r=O(1)$ as $s\xrightarrow[]{}-\infty.$ This implies that: \[\theta\xrightarrow[]{}\pm\sqrt{3}\text{ as }s\xrightarrow[]{}-\infty.\]

We now study the phase portrait of the autonomous system (\ref{system}) and look for critical points satisfying the needed asymptotics. We have the following options:
\begin{itemize}
    \item The origin of the $(\theta,\alpha)$ plane. This point was used in \cite{nakedsingularities} as the starting point for all the orbits, since it is the only critical point satisfying the required regularity at the center for $(3+1)$-dimensional solutions. However, this is the wrong asymptotic condition for $\theta$.
    \item $P_{1}$ is a critical point with $\theta=1.$ This point was used in \cite{nakedsingularities} as the endpoint for the orbits corresponding to solutions with null singular boundary. However, this is the wrong asymptotic condition for $\theta$.
    \item $P_{-1}$ is a critical point with $\theta=-1,$ with similar properties to $P_{1}$. Again, this has the wrong asymptotic condition for $\theta$.
    \item $P_{0}$ is a critical point with $\theta=-k,\ \alpha=1.$ This point was used in \cite{nakedsingularities} as the endpoint for the orbits corresponding to solutions with naked singularities. However, this point is unstable to the past.
    \item $P_{\pm}$ is the critical point with $\theta=-\frac{1}{k}=\pm\sqrt{3}$ and $\alpha=0.$ This satisfies the right asymptotic condition for $\theta$.
     \item $Q_{\pm}$ is the critical point with $\theta=\mp\frac{1}{\sqrt{3}},\ \alpha=0,$ but again this has the wrong asymptotic condition for $\theta$.
     \item The orbit could have the asymptotic line $\theta=\pm\sqrt{3}$. In this case, $(0,\pm\sqrt{3})$ would be a critical point for the $(\beta,\theta)$ system, but that is false.
\end{itemize}

Based on this discussion, we remark in order for the $(4+1)$-dimensional vacuum solution to be regular at $\Gamma$, the orbits must start from $P_{\pm}.$ From now on, we use this as the boundary condition at the center for all our Einstein-Scalar field solutions.

We notice that we could study self-similar Einstein-Scalar field solutions with a general $k.$ In that case, the corresponding vacuum metrics are not necessarily scale-invariant, but they satisfy a similar more general condition. However, the regularity condition at the center $\theta=\pm\sqrt{3}$ implies that $k=\mp\frac{1}{\sqrt{3}}.$ Thus, we recover the scale invariance property.

\subsubsection{Expansion of the solutions near the center}\label{expansion near center section}
We now study the system $(\ref{system})$ near the negatively stable critical point $(\theta,\alpha)=(\pm\sqrt{3},0)=P_{\pm}.$ Using the formulas $(\ref{zeta})$ and $(\ref{e{2lambda}}),$ we have:
\[\zeta\xrightarrow[]{}\pm\sqrt{3},\ e^{2\lambda}\xrightarrow[]{}0,\ \mu\xrightarrow[]{}-\infty\text{ as }s\xrightarrow[]{}-\infty.\]

We compute the expansion of the solutions near the center. We briefly recall the results in Section \ref{Appendix Expansion Subsection}. We define $\mathfrak{a}$ to be the set of equivalence classes of the relation $(a_1,a_2)\sim(Aa_1,A^2a_2),$ for all $A>0.$ For any orbit of $(\ref{system})$ near $P_{\pm}$, there exists a unique pair $(a_1,a_2)\in\mathfrak{a}$ such that:
\[\alpha=a_1e^{2s}-a_1^2e^{4s}+O(e^{6s}),\]
\[\theta=\pm\sqrt{3}\mp\frac{a_1}{\sqrt{3}}e^{2s}+\bigg(\pm\frac{a_1^2}{\sqrt{3}}+a_1^2d^{11}+a_2\bigg)e^{4s}+O(e^{6s}).\]
We also recall that $d^{11}$ is a constant depending on the system $(\ref{system})$. If $a_1=0$ we get that $\alpha\equiv 0,$ so we assume from now on that $a_1\neq0.$ Using $(\ref{zeta}),\ (\ref{e{2lambda}}),$ we compute further:
\[\beta=\frac{1}{a_1}e^{-2s}+1+O(e^{2s}),\]
\[e^{2\lambda}=\bigg(-\frac{a_1^2}{3}\mp a_1^2d^{11}\mp a_2\bigg)e^{4s}+O(e^{6s}),\]
\[e^{\nu-\lambda}=-\frac{2}{a_1}e^{-s}+O(e^{3s}),\]
\[\zeta=\pm\sqrt{3}\pm\frac{a_1}{\sqrt{3}}e^{2s}+\big(a_1^2d^{11}+a_2\big)e^{4s}+O(e^{6s}).\]
We remark that as a consequence of Assumption \ref{assumption non trapped}, we must have:
\[a_1<0\text{ and }\mp d^{11}\mp \frac{a_2}{a_1^2}>\frac{1}{3}.\]
It is convenient to introduce the parameter $X=\mp d^{11}\mp \frac{a_2}{a_1^2}-\frac{1}{3}>0.$ We proved the following:

\begin{proposition}\label{expansion near center proposition}
    Consider any $(4+1)$-dimensional scale-invariant vacuum solution with an $SO(3)\times U(1)$ isometry. Suppose that it also satisfies Assumptions \ref{assumption}, \ref{assumption non trapped} and \ref{assumption scale}. Then $(\theta,\alpha)|_{\Gamma}=(\pm\sqrt{3},0),$ and there exists $X>0$ such that near the center we have $\alpha<0$ and the expansion:
    \begin{equation}\label{theta in terms of alpha}
    \theta=\pm\sqrt{3}\mp\frac{\alpha}{\sqrt{3}}\mp\bigg(X+\frac{1}{3}\bigg)\alpha^2+O(|\alpha|^3).
\end{equation}
Moreover, for every $X>0$ there exists at most one such solution given in self-similar Bondi gauge, up to the gauge freedom (\ref{gauge freedom}).
\end{proposition}

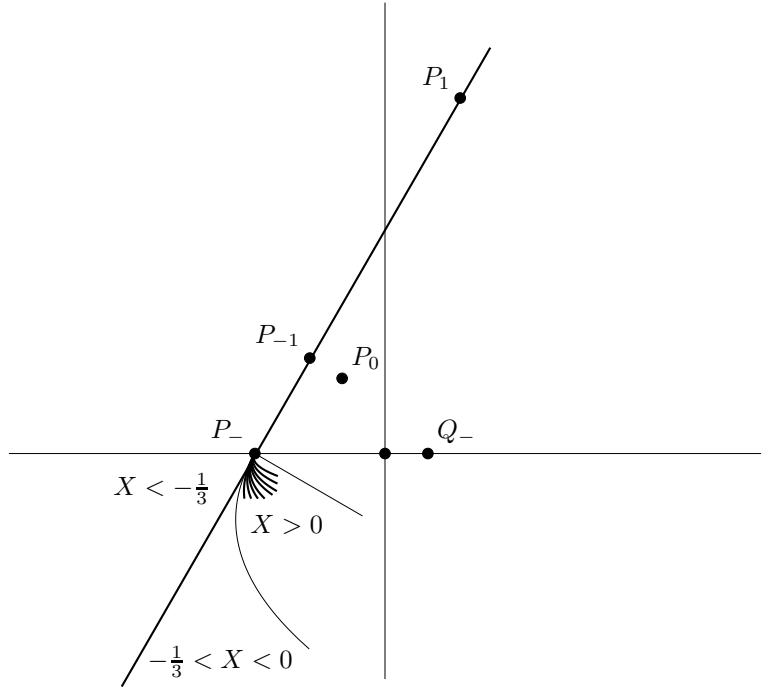
\begin{center}
\begin{tikzpicture}
\draw[very thin] (-5,0) -- (5,0);
\draw[very thin] (0,-3) -- (0,6);
\filldraw[color=black, fill=black] (0,0) circle (2pt);
\filldraw[color=black, fill=black] (-1.73,0) circle (2pt) node[anchor=south east] {$P_-$};
\filldraw[color=black, fill=black] (0.57,0) circle (2pt) node[anchor=south west] {$Q_-$};
\filldraw[color=black, fill=black] (-0.57,1) circle (2pt) node[anchor=south west] {$P_0$};
\filldraw[color=black, fill=black] (-1,1.27) circle (2pt) node[anchor=south east] {$P_{-1}$};
\filldraw[color=black, fill=black] (1,4.73) circle (2pt) node[anchor=south east] {$P_{1}$};
\draw[thick] (-3.5,-3.1) -- (1.4,5.4);
\draw[very thin] (-1.73,0) -- (-0.3,-0.83);
\draw[very thin]   plot[smooth,domain=-2.6:0] ({-1.73+(0.58*\x)+(0.33*(\x)^2)},\x);
\draw[thick]   plot[smooth,domain=-0.6:0] ({-1.73+(0.58*\x)+((1+0.33)*(\x)^2)},\x);
\draw[thick]   plot[smooth,domain=-0.6:0] ({-1.73+(0.58*\x)+((0.5+0.33)*(\x)^2)},\x);
\draw[thick]   plot[smooth,domain=-0.6:0] ({-1.73+(0.58*\x)+((0.25+0.33)*(\x)^2)},\x);
\draw[thick]   plot[smooth,domain=-0.6:0] ({-1.73+(0.58*\x)+((0.75+0.33)*(\x)^2)},\x);
\draw[thick]   plot[smooth,domain=-0.55:0] ({-1.73+(0.58*\x)+((1.5+0.33)*(\x)^2)},\x);
\draw[thick]   plot[smooth,domain=-0.5:0] ({-1.73+(0.58*\x)+((2+0.33)*(\x)^2)},\x);
\draw[thick]   plot[smooth,domain=-0.4:0] ({-1.73+(0.58*\x)+((3+0.33)*(\x)^2)},\x);
\draw[thick]   plot[smooth,domain=-0.3:0] ({-1.73+(0.58*\x)+((5+0.33)*(\x)^2)},\x);
\draw (-2.2,-0.8)  node[anchor=south east] {$X<-\frac{1}{3}$};
\draw (-1.1,-3.1)  node[anchor=south east] {$-\frac{1}{3}<X<0$};
\draw (-0.7,-1.2)  node[anchor=south east] {$X>0$};
\end{tikzpicture}
\end{center}

\subsubsection{Cone interior solutions}\label{interior solutions section}
The general vacuum solutions considered above are defined in a neighbourhood of the center. By scale invariance, they can be extended all the way to the scaling origin along lines of constant $x$. We prove that when restricted to the maximal interval of existence of $(\ref{system}),$ they are cone interior solutions, in the sense that they are defined in the chronological past of the scaling origin. We shall see later that the scaling origin represents the first singularity at the center.

\begin{lemma}
 For any solution to $(\ref{system})$ defined on the maximal interval of existence, we have $\alpha<0.$
\end{lemma}
\begin{proof}
We know $\alpha\xrightarrow[]{}0-$ as $s\xrightarrow[]{}-\infty.$ Since $k^2=\frac{1}{3},$ we get from $(\ref{system})$ that $\frac{d\alpha}{ds}<0$ and we conclude.
\end{proof}

\begin{proposition}\label{general interior solution proposition}
    Consider any vacuum spacetime satisfying the hypothesis of Proposition \ref{expansion near center proposition} and restrict it to the maximal interval of existence of $(\ref{system}).$ Then it is a cone interior solution.
\end{proposition}
\begin{proof}
Take any $x\geq0$, such that the restricted solution is defined for all $(u,r)$ with $x=-r/u$. Since $\beta(x)<0,$ we obtain $\Tilde{g}(-S,-S)=2xu^2e^{\nu+\lambda}\beta f^{-1}<0.$ Thus, the curve $\gamma_x=\{x=\text{const}\}$ is future directed timelike.

Define another null coordinate $v$ which is constant along incoming null cones, such that $v=u$ on $\Gamma$ and $v\in(u,v_X(u))$ on $C^+_u\cap\{x<X\}.$ The causality of the spacetime induced on $\mathcal{U}^{(1+1)}$ is the same as that of $\mathbb{R}^{1+1}_{(u,v)}.$ Thus, each $\gamma_x$ is a future directed timelike curve passing through the scaling origin $(0,0)$ in $\mathbb{R}^{1+1}_{(u,v)}.$ We conclude that all points with $\beta(x)<0$ are in the chronological past of the scaling origin, so the restricted solution is a cone interior solution.
\end{proof}

For solutions that can be extended beyond the cone interior region, we use the above double null coordinates to define $C_0^-$ as the boundary of the causal past of the scaling origin. We obtain:

\begin{corollary}\label{reach C0- condition}
If the spacetime can be extended to $\beta>0,$ then $C_0^-=\{\beta=0\}$.
\end{corollary}
\begin{proof}
Follows because on $\{\beta=0\}$ the scaling vectorfield $S$ is null.
\end{proof}

This result provides a criterion to check whether a given vacuum spacetime satisfying our assumptions is defined beyond the cone interior region. 

Alternatively, we can start with an orbit of (\ref{system}) satisfying the expansions in Section \ref{expansion near center section}, and consider the corresponding vacuum solution. In order for this to extend beyond the cone interior region, we must have $\alpha\xrightarrow[]{}-\infty$ within finite parameter $s.$ Therefore, the goal is to prove that the solution to (\ref{system}) blows up in finite time. In \cite{nakedsingularities}, such a result is proved by a qualitative analysis of the relevant orbit.  In our case, this is achieved by using the explicit Lorentzian Hawking--Page solutions, as we shall see in the next section.

\subsection{Lorentzian Hawking--Page Solutions}\label{Interior Hawking--Page Solutions section}
We recall that we constructed the Lorentzian Hawking--Page solutions on $\mathcal{U}^{(4+1)}=(0,\infty)\times(\rho_+,\infty)\times S^2\times S^1$:
\[g^{(4+1)}_m=-d\tau^2+\tau^2V^{-1}d\rho^2+\tau^2\rho^2d\sigma_2^2+\bigg(\frac{b}{2\pi}\bigg)^2\tau^2Vd\sigma_1^2.\]
Moreover, we proved they extend to a regular solution on $\mathcal{U}^{(4+1)}\cup\Gamma^{(4+1)}$. Also, by the Kaluza--Klein reduction, we constructed the Hawking--Page Einstein-Scalar field solutions on $\mathcal{U}^{(3+1)}=(0,\infty)\times(\rho_+,\infty)\times S^2:$
\begin{equation}\label{HP solutions metric}
    g=-\frac{b}{2\pi}\tau\sqrt{V}d\tau^2+\frac{b}{2\pi}\frac{\tau^3}{\sqrt{V}}d\rho^2+r^2d\sigma_2^2
\end{equation}
\begin{equation}\label{HP solutions rho}
    r=\sqrt{\frac{b}{2\pi}}\tau^{3/2}V^{1/4}\rho
\end{equation}
\begin{equation}\label{HP solutions phi}
    \phi=\pm\frac{\sqrt{3}}{2}\log\bigg(\frac{b}{2\pi}\tau\sqrt{V}\bigg).
\end{equation}

In this section, we write the Lorentzian Hawking--Page solutions in self-similar Bondi gauge, in order to use the general theory developed in Section \ref{general self similar solutions}. In particular, we can view each Lorentzian Hawking--Page solution as an orbit of the autonomous system $(\ref{system})$, with its future Cauchy horizon given by a singular point of the orbit. This will be essential in order to extend the solutions to the cone exterior region. We will also obtain the interesting fact that all $(4+1)$-dimensional scale-invariant vacuum solutions with an $SO(3)\times U(1)$ isometry, satisfying Assumptions \ref{assumption}, \ref{assumption non trapped} and \ref{assumption scale} are given by a Lorentzian Hawking--Page solution in the cone interior region.

\subsubsection{Lorentzian Hawking--Page solutions in self-similar Bondi gauge}\label{Hawking--Page solutions in self-similar Bondi gauge}
We write the Lorentzian Hawking--Page solutions in self-similar Bondi gauge. We define the outgoing null coordinate:
\[U(\tau,\rho)=-\tau\exp\bigg(\int_{\rho_+}^\rho\frac{d\tilde{\rho}}{\sqrt{V(\tilde{\rho})}}\bigg).\]
For notation purposes, it is convenient to introduce the function $F(\rho)$ such that $\tau=-UF.$ We compute: 
\[\partial_U=-F\partial_{\tau},\ \partial_{\rho}=\partial_{\rho}+\frac{UF}{\sqrt{V}}\partial_{\tau}.\]
Therefore, in $(U,\rho)$ coordinates we have:
\[S=\frac{2}{3}U\partial_U,\ r=\sqrt{\frac{b}{2\pi}}(-U)^{3/2}F^{3/2}V^{1/4}\rho.\]
\[g=\frac{b}{2\pi}UF^3\sqrt{V}dU^2-2\frac{b}{2\pi}U^2F^3dUd\rho+r^2d\sigma_2^2.\]
We notice that $\partial_U$ is a future directed timelike vectorfield. We still have the gauge freedom $U=U(u)$, so we make the change of coordinates:
\[-u=(-U)^{3/2}.\]
In $(u,\rho)$ coordinates we have:
\[S=u\partial_u,\ r=-\sqrt{\frac{b}{2\pi}}uF^{3/2}V^{1/4}\rho,\]
\[g=-\frac{b}{2\pi}\frac{4}{9}F^3\sqrt{V}du^2+2\frac{b}{2\pi}\frac{2}{3}uF^3dud\rho+r^2d\sigma_2^2.\]
The desired self-similar Bondi coordinates on $\mathcal{U}^{(3+1)}$ are $(u,r),$ since we have:
\[S=u\partial_u+r\partial_{r}.\]
The region $\mathcal{U}^{(3+1)}$ is covered by the coordinate patch $(\tau,\rho,\theta,\varphi),$ so we have the explicit formulas:
\begin{equation}\label{HP in Bondi coordinates}
    g=-\bigg(\frac{b}{2\pi}\frac{4}{9}F^3\sqrt{V}+\frac{b}{2\pi}\frac{4}{3}\frac{1}{A}F^3F^{3/2}V^{1/4}\rho\bigg)du^2-2\bigg(\sqrt{\frac{b}{2\pi}}\frac{2}{3}F^3\frac{1}{A}\bigg)dudr+r^2d\sigma_2^2
\end{equation}
\begin{equation}
    A=-F^{3/2}V^{-1/4}\rho\bigg(\frac{3}{2}-\frac{\sqrt{V}}{\rho}-\frac{\dot{V}}{4\sqrt{V}}\bigg)
\end{equation}
\begin{equation}\label{interior HP phi formula}
    \phi=\pm\frac{\sqrt{3}}{2}\log\bigg(\frac{b}{2\pi}F\sqrt{V}\bigg)-k\log(-u).
\end{equation}
Using the definition of $V$, we can readily compute that $A>0$ for $\rho\in(\rho_+,\infty).$ We can now identify $e^{2\nu}$ and $e^{\nu+\lambda}$ in the above formula, so we also compute:
\begin{equation}\label{HP beta formula}
    \beta=\frac{1}{3}\bigg(\frac{3}{2}-\frac{\sqrt{V}}{\rho}-\frac{\dot{V}}{4\sqrt{V}}\bigg)
\end{equation}
\begin{equation}\label{HP theta formula}
    \theta=\frac{\pm\sqrt{3}\rho(\dot{V}-2\sqrt{V})}{4V+\dot{V}\rho-6\sqrt{V}\rho}
\end{equation}
The second equation follows because:
\[\theta=r\partial_{r}\phi=\frac{F^{3/2}V^{1/4}\rho}{A}\bigg(\partial_{\rho}\phi-\frac{\tau}{\sqrt{V}}\partial_{\tau}\phi\bigg).\]

\subsubsection{Cone interior solutions}

We expressed the Lorentzian Hawking--Page solutions in self-similar Bondi gauge, so from now on we identify them with the corresponding orbits of the autonomous system $(\ref{system})$. We begin by computing the radius function along outgoing null cones when the coordinates $(\tau,\rho)$ become singular. Since $V\sim \rho^2$ as $\rho\xrightarrow[]{}\infty,$ we have:
\[\lim_{\rho\xrightarrow[]{}\infty}\rho\exp\bigg(-\int_{\rho_+}^\rho\frac{d\tilde{\rho}}{\sqrt{V(\tilde{\rho})}}\bigg)=C_1(m),\]
for some constant $C_1$ depending on $m$. Thus, we get that on the outgoing null cone of constant $u$ we have:
\[r_*(u):=\lim_{\rho\xrightarrow[]{}\infty}r(u,\rho)=-\sqrt{\frac{b}{2\pi}}u\big(C_1(m)\big)^{\frac{3}{2}}.\]
As in the general setting, we introduce the coordinate:
\[e^s=-\frac{r}{u},\]
and denote $s_*=s(r_*(u),u).$ We remark that $r\xrightarrow[]{}0$ as $\rho\xrightarrow[]{}\rho_+,$ so in $(u,s)$ coordinates the center $\Gamma^{(4+1)}$ is given by $s\xrightarrow[]{}-\infty$. Moreover, the Lorentzian Hawking--Page solutions on $\mathcal{U}^{(4+1)}$ are defined for $s\in(-\infty,s_*),$ so $\{s=s_*\}$ represents their future Cauchy horizon.

\begin{proposition}
    The Lorentzian Hawking--Page solutions on $\mathcal{U}^{(4+1)}\cup\Gamma^{(4+1)}$ are cone interior solutions. Moreover, in $(u,s)$ coordinates we have $\beta\xrightarrow[]{}0$ and $\theta\xrightarrow[]{}\mp\sqrt{3}=\frac{1}{k}$ as $s\xrightarrow[]{}s_*.$
\end{proposition}
\begin{proof}
    It is immediate to show that $\beta(\rho)<0$ on $(\rho_+,\infty)$, so by Proposition \ref{general interior solution proposition} we obtain that they are cone interior solutions. We also compute the expansion as $\rho\xrightarrow[]{}\infty$:
    \[\sqrt{V}=\rho+\frac{1}{2\rho}-\frac{m}{\rho^2}+O(\rho^{-3}).\]
    By formulas (\ref{HP beta formula}) and (\ref{HP theta formula}), we get that $\lim_{\rho\xrightarrow[]{}\infty}\beta(\rho)=0$ and $\lim_{\rho\xrightarrow[]{}\infty}\theta(\rho)=\mp\sqrt{3}.$
\end{proof}

\begin{corollary}\label{HP solutions interior behaviour}
Consider any Lorentzian Hawking--Page solution. In the $(\theta,\beta)$ plane the corresponding orbit approaches $(\frac{1}{k},0)$ as $s\xrightarrow[]{}s_*$. In the $(\theta,\alpha)$ plane the orbit starts at $(\pm\sqrt{3},0)$ as $s\xrightarrow[]{}-\infty$ and has a vertical asymptote at $\theta=\frac{1}{k}$ in the lower half plane as $s\xrightarrow[]{}s_*$.
\end{corollary}

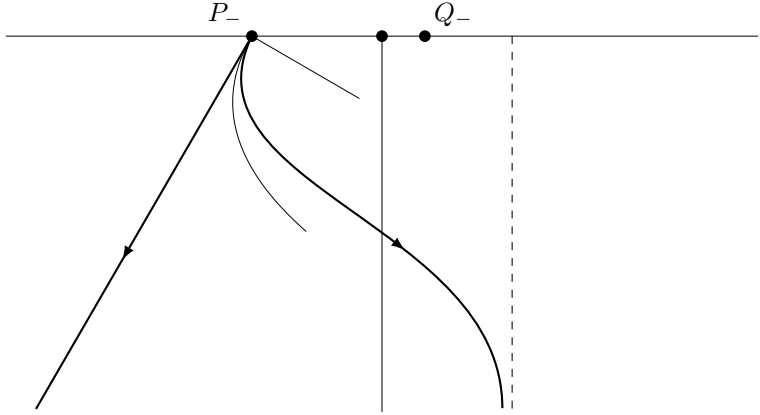
\begin{center}
\begin{tikzpicture}[decoration={markings, mark=at position 0.6 with {\arrow{latex}}}] 
\draw[dashed] (1.73,0) -- (1.73,-5);
\draw[thick][postaction={decorate}] (-1.73,0) .. controls (-2.7,-2) and (1.6,-2.5) ..(1.6,-4.95);
\draw[very thin] (-5,0) -- (5,0);
\draw[very thin] (0,-5) -- (0,0);
\filldraw[color=black, fill=black] (0,0) circle (2pt);
\filldraw[color=black, fill=black] (-1.73,0) circle (2pt) node[anchor=south east] {$P_-$};
\filldraw[color=black, fill=black] (0.57,0) circle (2pt) node[anchor=south west] {$Q_-$};
\draw[thick][postaction={decorate}] (-1.73,0) -- (-4.6,-4.96);
\draw[very thin] (-1.73,0) -- (-0.3,-0.83);
\draw[very thin]   plot[smooth,domain=-2.6:0] ({-1.73+(0.58*\x)+(0.33*(\x)^2)},\x);
\end{tikzpicture}
\end{center}

We proved that the orbits of (\ref{system}) corresponding to the Lorentzian Hawking--Page solutions blow up in finite time, with $\alpha\xrightarrow[]{}-\infty.$ Based on our discussion at the end of Section \ref{interior solutions section}, this implies that each solution is defined everywhere in the chronological past of the scaling origin. Thus, the future Cauchy horizon of each Lorentzian Hawking--Page solution corresponds to the singular point of its orbit. In the next section, we continue the orbits beyond the singular point, so we extend the solutions to the cone exterior region.

We conclude the section by showing that the Lorentzian Hawking--Page solutions are the only possible interior solutions to the general problem considered in Section \ref{general self similar solutions}:

\begin{proposition}\label{interior all solutions}
     Consider a $(4+1)$-dimensional scale-invariant vacuum solution with an $SO(3)\times U(1)$ isometry, which satisfies Assumptions \ref{assumption}, \ref{assumption non trapped} and \ref{assumption scale}. Let $X$ be the corresponding parameter from Proposition \ref{expansion near center proposition}. Then, up to the gauge freedom (\ref{gauge freedom}), the solution is the Lorentzian Hawking--Page solution with:
    \[9\rho_+^2=\frac{1}{X}.\]
\end{proposition}
\begin{proof}
    By the expansions near the center in Section \ref{expansion near center section}, we know that:
    \[X=\lim_{s\xrightarrow[]{}-\infty}\frac{e^{2\lambda}}{\alpha^2}.\]
    For Lorentzian Hawking--Page solutions, we compute:
    \[\lim_{\rho\xrightarrow[]{}\rho_+}\frac{e^{2\lambda}}{\alpha^2}=\lim_{\rho\xrightarrow[]{}\rho_+}\frac{(e^{\nu+\lambda})^2}{e^{2\nu}}\cdot\beta^2=\frac{1}{9\rho_+^2}.\]
    We notice that $\rho_+\xrightarrow[]{}0$ as $m\xrightarrow[]{}0$ and $\rho_+\xrightarrow[]{}\infty$ as $m\xrightarrow[]{}\infty,$ so for every $X>0$ there exists $m>0$ such that $9\rho_+^2=1/X.$ By the uniqueness result in Proposition \ref{expansion near center proposition} and the regularity in Remark \ref{conical singularity remark}, we conclude that the solution is given by the Lorentzian Hawking--Page solution with $9\rho_+^2=1/X$, up to the gauge freedom (\ref{gauge freedom}).
\end{proof}
\begin{remark}
    Given an orbit of (\ref{system}) satisfying the expansions in Section \ref{expansion near center section}, we obtain a $(4+1)$-dimensional scale-invariant vacuum solution with an $SO(3)\times U(1)$ isometry, which satisfies Assumptions \ref{assumption non trapped} and \ref{assumption scale}. According to the above result, there exists a transformation $\phi\rightarrow\phi+C$ such that the spacetime is precisely a Lorentzian Hawking--Page solution, up to the gauge freedom (\ref{gauge freedom}). In particular, the spacetime is regular at the center and it also satisfies Assumption \ref{assumption}. We notice that we obtained regularity at the center by classifying all the possible interior solutions.
\end{remark}

\section{Extensions of Lorentzian Hawking--Page Solutions}\label{second big section}

In this section we prove the main result of the paper by constructing all the maximal extensions of the Lorentzian Hawking--Page solutions in the class of $(4+1)$-dimensional scale-invariant vacuum spacetimes with an $SO(3)\times U(1)$ isometry. We begin with an analysis of the singular point of the orbits of (\ref{system}), which corresponds to the future Cauchy horizon of the Lorentzian Hawking--Page solutions. We prove that each interior Lorentzian Hawking--Page solution bifurcates to a one-parameter family of exterior solutions with a loss of regularity across the cone. We then consider the extensions in the regular non-trapped region. In the case of extensions to spacetimes with a null curvature singularity, the regular region represents the maximal extension of the interior Lorentzian Hawking--Page solutions. However, for the other extensions we need to study the trapped region in order to construct spacetimes with a spacelike singularity and spacetimes with a null Cauchy horizon of Taub--NUT type.

\subsection{Extending Beyond the Cone Interior Region}\label{Extending Beyond the Interior Region section}

We consider a fixed Lorentzian Hawking--Page spacetime. According to Corollary \ref{HP solutions interior behaviour}, the corresponding orbit of (\ref{system}) blows up with $\theta\rightarrow\pm\sqrt{3},\ \alpha\rightarrow-\infty$ at the future Cauchy horizon $C_0^-=\{s=s_*\}$. We can continue the orbit beyond this singular point by using the methods of Christodoulou in \cite{nakedsingularities}. This corresponds to extending the original spacetime into the cone exterior region. 

We recall that in the $(\theta,\beta)$ plane the autonomous system takes the form:

\begin{equation}\label{system theta beta}
    \begin{dcases}
    \frac{d\theta}{ds}=\frac{k}{\beta}\big(k\theta-1\big)+\theta\big[(\theta+k)^2-(1+k^2)\big] \\
    \frac{d\beta}{ds}=1-k^2-\beta\big[(\theta+k)^2+(1-k^2)\big]
    \end{dcases}
\end{equation}

The solution blows up at the point $(\frac{1}{k},0)$ as $s\xrightarrow[]{}s_*$. To deal with this issue, we define the new variable $t$ by:
\[\frac{ds}{dt}=\beta.\]
Using (\ref{system theta beta}), we note that: \[\frac{\beta}{s-s_*}\xrightarrow[]{}1-k^2\text{ as }s\xrightarrow[]{}s_*.\] This implies that $t\xrightarrow[]{}-\infty$ as $s\xrightarrow[]{}s_*$. In terms of the $t$ variable, the system becomes:
\begin{equation}\label{system in t variable}
    \begin{dcases}
    \frac{d\theta}{dt}=k\big(k\theta-1\big)+\beta\theta\big[(\theta+k)^2-(1+k^2)\big] \\
    \frac{d\beta}{dt}=\beta(1-k^2)-\beta^2\big[(\theta+k)^2+(1-k^2)\big]
    \end{dcases}
\end{equation}

We remark that the point $(\frac{1}{k},0)$, which was a singular point for (\ref{system theta beta}), is a negatively stable critical point for (\ref{system in t variable}). Moreover, there exists an orbit that corresponds to the interior Lorentzian Hawking--Page solution considered above. We could repeat the argument using the Poincaré-Dulac theorem as in Appendix \ref{Appendix Expansion Subsection} to describe the behavior of the solutions in a neighborhood of the critical point. However, this was already carried out in \cite{nakedsingularities}, so we quote directly their result. There exists a parameter $a_1^{(i)},$ such that for $s<s_*$ the interior Lorentzian Hawking--Page solution has the expansion:
\[\beta=\frac{2}{3}(s-s_*)+O(|s-s_*|^2),\]
\[\theta-\frac{1}{k}=\frac{12}{k}\beta+a^{(i)}_1\big(s-s_*\big)^{\frac{1}{2}}+O(|s-s_*|).\]
We point out that according to \cite{nakedsingularities}, the parameter $a_1^{(i)}=0$ corresponds to an explicit analytic solution of the Einstein-Scalar field equations. However, the associated orbit of (\ref{system}) starts at the origin as $s\rightarrow-\infty.$ By our results in Section \ref{Boundary condition at the center}, we have $a_1^{(i)}\neq0$ for any Lorentzian Hawking--Page solution.

In order to extend the interior Lorentzian Hawking--Page solution, we are free to choose any orbit of (\ref{system in t variable}) which converges to $(\frac{1}{k},0)$ as $t\rightarrow-\infty$ and satisfies $\beta>0$. As a result, the cone interior solution bifurcates to a 1-parameter family of solutions. For any parameter $a_1^{(e)}$, we can extend to a solution in the region $s>s_*$ with expansion:
\[\beta=\frac{2}{3}(s-s_*)+O(|s-s_*|^2),\]
\[\theta-\frac{1}{k}=\frac{12}{k}\beta+a^{(e)}_1\big(s_*-s\big)^{\frac{1}{2}}+O(|s-s_*|).\]
This determines the metric $g$ in self-similar Bondi coordinates on $s>s_*.$ To compute $\phi$ starting from $\theta,$ we use the boundary values induced by (\ref{interior HP phi formula}) at $s=s_*.$

We remark that according to Corollary \ref{reach C0- condition}, in the extended solution we have $C_0^-=\{\beta=0\}=\{s=s_*\}$. Thus, the original Lorentzian Hawking--Page solution defined in Section \ref{Interior Hawking--Page Solutions section} represents the cone interior solution in the extended spacetime. 

The solution is smooth away from the cone $C_0^-.$ The regularity at $s=s_*$ is determined by the above expansions. In Bondi coordinates we have that $\phi\in C^{1,1/2}$ along outgoing cones of constant $u$. Thus, we also have $f\in C^{1,1/2},$ which represents the area radius of the spheres $S^1$, so the vacuum metric has the same regularity. We cannot expect to have better regularity, since $a_1^{(i)}\neq0$ for our solutions. 

\begin{remark}
    The problem of extending Lorentzian cone metrics to the cone exterior region was also addressed by Anderson in \cite{Andersoncone} and \cite{Anderson}. The argument used the Fefferman-Graham expansion of \cite{Fefferman-Graham} and a form of Wick rotation. We consider the geodesic compactification of the Hawking--Page solutions $g_m^{(4)}$ to be $\tilde{g}_m^{(4)}=R^2g_m^{(4)},$ where $R(p)$ is the distance to the boundary with respect to the conformal metric. We have:
    \[\tilde{g}_m^{(4)}=dR^2+g_R,\]
    where $g_{R}$ are metrics on the conformal boundary $S^2\times S^1.$ We compute the Fefferman-Graham expansion:
    \[g_R=\big(C_1(m)\big)^2\cdot\bigg(1-\frac{1}{2}R^2+\frac{2m}{3}R^3+\cdots\bigg)d\sigma_2^2+\bigg(C_1(m)\frac{b}{2\pi}\bigg)^2\cdot\bigg(1+\frac{1}{2}R^2-\frac{4m}{3}R^3+\cdots\bigg)d\sigma_1^2.\]
    We recall that the Lorentzian Hawking--Page solutions are given by:
    \[g_m^{(4+1)}=-d\tau^2+\frac{\tau^2}{R^2}\big(dR^2+g_R\big).\]
    According to \cite{Anderson}, each solution has a regular Fefferman-Graham expansion, but the metrics are not smooth across the cone $C_0^-$ in this case because of the presence of the cubic terms in the expansions. The regularity obtained by this argument for the extended metric across the cone $C_0^-$ is $C^{1,1/2}$, as in our above proof.
\end{remark}

We briefly argue that our extended solutions are uniquely determined by characteristic initial data induced on an outgoing null cone $C_0^+.$ The cone interior solution has smooth initial data, so it is uniquely determined. We fix some $s_0<s_*$ such that $m>0$ on the timelike curve $\{s=s_0\}.$ For any $s_1\in(s_0,s_*),$ we denote by $C_{[s_0,s_1]}^-$ the incoming null cone starting at $C_0^+\cap\{s=s_1\},$ restricted to $s\in[s_0,s_1].$ We consider the initial data for the Einstein-Scalar field equations induced on $C_{[s_0,s_1]}^-\cup C_0^+.$ We recall the extension criterion of \cite{kommemi}, according to which we have $r=0$ at a first singularity away from the center. We point out that we can apply the extension criterion even though we have $\theta\in C^{0,1/2}$ on $C_0^+\big|_{s\geq s_0}$, because of the BV theory developed in \cite{BV}. Using this we get that our scale-invariant solution is uniquely determined, as long as $r\neq 0$ on $C_0^+.$ The same extension criterion can be used to argue in the case of non-characteristic initial data. 

We summarize the results proved so far in this section:

\begin{proposition}
    Each interior Lorentzian Hawking--Page solution can be extended in a neighborhood of $C_0^-$ in the cone exterior region to a 1-parameter family of scale-invariant globally hyperbolic vacuum solutions with an $SO(3)\times U(1)$ isometry. The solutions are smooth away from $C_0^-$, while across $C_0^-$ the regularity is $C^{1,1/2}$. Moreover, they are uniquely determined by characteristic initial data induced on an outgoing null cone.
\end{proposition}

We denote the cone exterior region by $\mathcal{E}^{(4+1)}.$ So far, this is only defined in a neighborhood to the future of $C_0^-$. We denote the entire manifold by $\mathcal{M}^{(4+1)}=\Gamma^{(4+1)}\cup\mathcal{U}^{(4+1)}\cup\mathcal{E}^{(4+1)},$ and the two dimensional Lorentzian quotient manifold by $\mathcal{Q}^{(1+1)}=\big(\mathcal{U}^{(4+1)}\cup\mathcal{E}^{(4+1)}\big)/(SO(3)\times U(1)).$ We recall that $\big(\mathcal{Q}^{(1+1)}\times S^2,g,\phi\big)$ is the corresponding solution of the Einstein-Scalar field equations.

Finally, we prove that the vacuum solutions cannot be extended beyond $C_0^-$ near the center. As a result, the scaling origin represents the first singularity at the center:
\begin{proposition}
    For any extension of the Lorentzian Hawking--Page vacuum solution, the sectional curvature blows up along $C_0^-$.
\end{proposition}
\begin{proof}
    We notice that $\theta,\ \zeta,\ \mu$ are continuous across $C_0^-,$ so we have:
    \[\theta=\frac{1}{k},\ \zeta=-k,\ \frac{1}{1-\mu}=1+k^2,\ \text{ on }C_0^-.\]
    According to the computations in Appendix \ref{Appendix sectional curvature}, the sectional curvature of the 2-spheres is:
    \[K^{(4+1)}_{S^2(u,s_*)}=\frac{f}{r^2}\big(1-(1-\mu)(k\theta+1)(k\zeta+1)\big)=0.\]
    We compute the sectional curvature of $\mathcal{Q}^{(1+1)}$:
    \[K^{(4+1)}_{\mathcal{Q}^{(1+1)}}(u,s_*)=\frac{f}{r^2}\bigg[1-(1-\mu)(k\theta+1)(k\zeta+1)\bigg]+(1+k^2)\frac{f^2}{r^2}(1-\mu)\theta\zeta=-\frac{f^2}{r^2}.\]
    Since $\zeta=-k$ on $C_0^-,$ we get $\phi+k\log r=O(1)$ as $r\xrightarrow[]{}0.$ Thus, we have $f\sim r^{2/3}$ as $r\xrightarrow[]{}0.$ We conclude that $K^{(4+1)}_{\mathcal{Q}}(u,s_*)\xrightarrow[]{}-\infty$ as $r\xrightarrow[]{}0.$
\end{proof}

\subsubsection{Hawking--Page solutions and non-uniqueness}\label{non uniqueness}
According to \cite{Anderson}, the Hawking--Page solutions are examples of non-uniqueness to the problem of finding Einstein metrics with a given conformal type at infinity. We present this result in the context of Hawking--Page solutions and Lorentzian Hawking--Page solutions. The current section is independent of the rest of our paper and mostly follows \cite{Anderson}, but we include it to point out this interesting property of the Hawking--Page solutions.

We recall that the Hawking--Page metrics on $(\rho_+,\infty)\times S^2\times S^1$ are given by:
\[g^{(4)}_m=V^{-1}d\rho^2+\rho^2d\sigma_2^2+\bigg(\frac{b}{2\pi}\bigg)^2Vd\sigma_1^2,\]
where $m>0,\ V(\rho)=\rho^2+1-\frac{2m}{\rho},$ $\rho_+(m)$ is the largest root of $V$ and $b(m)=\frac{4\pi \rho_+}{3\rho_+^2+1}.$ 

The Hawking--Page metrics have conformal type at infinity:
\[d\sigma_2^2+\bigg(\frac{b}{2\pi}\bigg)^2d\sigma_1^2.\]
We define $m_0=\frac{2}{3\sqrt{3}}$ and compute $b(m_0)=\frac{1}{\sqrt{3}}.$ By studying the equation $3b\rho^2-2\rho+b=0,$ we get that for any $m\neq m_0,$ there exists $\tilde{m}\neq m$ such that $b(m)=b(\tilde{m}).$ In particular, the metrics $g^{(4)}_m$ and $g^{(4)}_{\tilde{m}}$ have the same conformal type at infinity. 

We consider the corresponding Lorentzian Hawking--Page solutions in self-similar Bondi gauge as in Section \ref{Hawking--Page solutions in self-similar Bondi gauge}. We compute on $C_0^-:$
\[r_m(u)=-u\sqrt{\frac{b(m)}{2\pi}}\big(C_1(m)\big)^{\frac{3}{2}},\ \phi_m(u)=-k\log\bigg(-u\bigg(\frac{b(m)}{2\pi}C_1(m)\bigg)^{\frac{3}{2}}\bigg).\]
Using the gauge freedom $u\rightarrow c^2u,$ we get $r_m=r_{\tilde{m}}$ and $\phi_m=\phi_{\tilde{m}}$ on $C_0^-.$ We also have $\beta(m)=\beta(\tilde{m}),\ \theta(m)=\theta(\tilde{m})=\frac{1}{k},\ e^{2\lambda}(m)=e^{2\lambda}(\tilde{m})=\frac{4}{3}$ on $C_0^-,$ so we get that the metrics $g^{(4+1)}_m$ and $g^{(4+1)}_{\tilde{m}}$ agree to first order on $C_0^-.$

\subsection{Exterior Solution in the Regular Region}

We define the regular region $\mathcal{R}^{(4+1)}=\{\mu<1\}\subset\mathcal{M}^{(4+1)}$ to be the region where the orbits $S^2\times S^1$ are not trapped with respect to $g^{(4+1)}$. Similarly, we introduce the notions of apparent horizon $\mathcal{A}^{(4+1)}=\{\mu=1\}\subset\mathcal{M}^{(4+1)}$ and trapped region $\mathcal{T}^{(4+1)}=\{\mu>1\}\subset\mathcal{M}^{(4+1)}.$  We notice that the regular region, apparent horizon and trapped region of $\big(\mathcal{Q}^{(1+1)}\times S^2,g\big)$ are then given by  $\mathcal{R}^{(3+1)}=\big(\mathcal{R}^{(4+1)}-\Gamma\big)/U(1),\ \mathcal{A}^{(3+1)}=\mathcal{A}^{(4+1)}/U(1),\ \mathcal{T}^{(3+1)}=\mathcal{T}^{(4+1)}/U(1).$ We remark that Bondi coordinates are well-defined in the regular region.

The exterior solutions $\big(\mathcal{E}^{(4+1)},g^{(4+1)}\big)$ are only defined in a neighborhood to the future of $C_0^-$ so far. In this section, we describe the maximal extension of the exterior solutions in the regular region. It turns out that the corresponding solutions of the Einstein-Scalar field equations coincide with the ones of \cite{nakedsingularities} in the cone exterior region. The $(3+1)$-dimensional spacetimes have a null curvature singularity or they extend up to the apparent horizon, which we interpret in the context of the $(4+1)$-dimensional vacuum solutions.

We study one of the exterior solutions obtained above in the $(\theta,\alpha)$ plane. We recall that it satisfies the autonomous system:

\begin{equation}\label{system alpha theta}
    \begin{dcases}
    \frac{d\theta}{ds}=k\alpha\big(k\theta-1\big)+\theta\big[(\theta+k)^2-(1+k^2)\big] \\
    \frac{d\alpha}{ds}=\alpha\big[(\theta+k)^2+(1-k^2)(1-\alpha)\big]
    \end{dcases}
\end{equation}
We assume $k=+\frac{1}{\sqrt{3}}$, and note the results in the other case are similar. The system has: the critical points stable to the future $P_0=(-k,1);$ the critical points unstable to the future $P_-=(-\sqrt{3},0),\ Q_-=(k,0);$ and the saddle critical points $(0,0),\ P_1=\big(1,\frac{2}{1-k}\big),\ P_{-1}=\big(-1,\frac{2}{1+k}\big).$ We also remark that the line $\alpha=\frac{\theta}{k}+3$ is an exact solution of the system, except at the critical points $P_{-},\ P_{-1}$ and $P_1.$ This property holds because of the parameter $k^2=\frac{1}{3}$.

Since $\alpha>0$ for $s>s_*,$ the orbit corresponding to our solution has a vertical asymptote at $\theta=\frac{1}{k}$ in the upper half plane as $s\xrightarrow[]{}s_*$. We cite the following result of Christodoulou from \cite{nakedsingularities}:

\begin{proposition}\label{Christodoulou phase portrait analysis}
    Consider orbits of (\ref{system alpha theta}) with an asymptote at $\theta=\frac{1}{k}$ in the upper half plane as $s\xrightarrow[]{}s_*+:$ 
    \begin{enumerate}
        \item  The upper branch $S_1^+$ of the stable manifold of $P_1$ represents such an orbit.
        \item For every orbit to the left of $S_1^+$ there exist $s_{\mathcal{A}}>s_*,\ \sigma_{\mathcal{A}}<0,$ such that the orbit blows up as $s\xrightarrow[]{}s_{\mathcal{A}}$ with $\theta\xrightarrow[]{}-\infty,\ \alpha\xrightarrow[]{}\infty,\  \alpha/\theta\xrightarrow[]{}\sigma_{\mathcal{A}},\ \alpha\sqrt{s_{\mathcal{A}}-s}\xrightarrow[]{}-\sigma_{\mathcal{A}}/\sqrt{2}$.
        \item For every orbit to the right of $S_1^+$ there exist $s_{\mathcal{A}}>s_*,\ \sigma_{\mathcal{A}}>0,$ such that the orbit blows up as $s\xrightarrow[]{}s_{\mathcal{A}}$ with $\theta\xrightarrow[]{}\infty,\ \alpha\xrightarrow[]{}\infty,\  \alpha/\theta\xrightarrow[]{}\sigma_{\mathcal{A}},\ \alpha\sqrt{s_{\mathcal{A}}-s}\xrightarrow[]{}\sigma_{\mathcal{A}}/\sqrt{2}$.
    \end{enumerate}
\end{proposition}

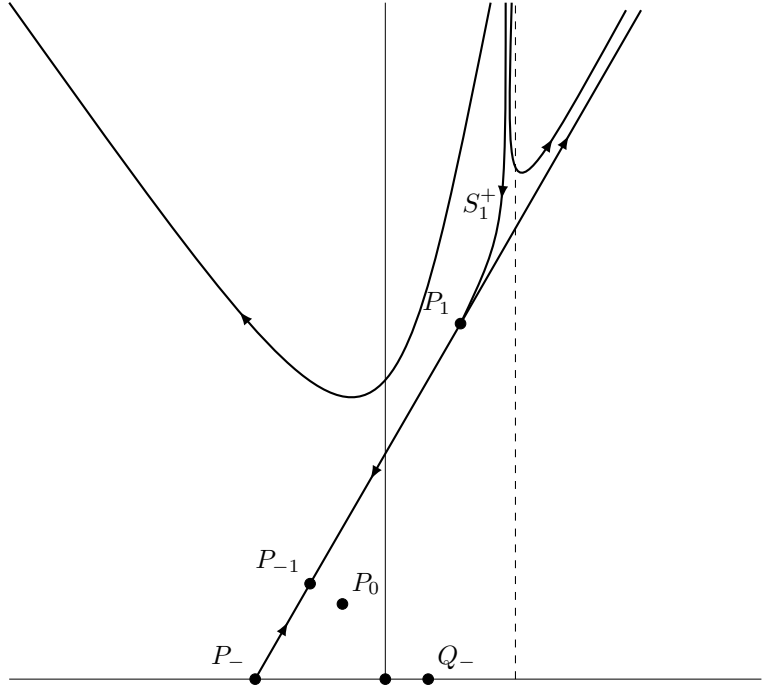
\begin{center}
\begin{tikzpicture}[decoration={markings, mark=at position 0.6 with {\arrow{latex}}}, scale=1] 
\draw[very thin] (-1.73,0) -- (-5,0);
\draw[very thin] (-1.73,0) -- (0,0);
\draw[very thin] (0.57,0) -- (0,0);
\draw[very thin] (0.57,0) -- (5,0);
\draw[very thin] (0,0) -- (0,9);
\filldraw[color=black, fill=black] (0,0) circle (2pt);
\filldraw[color=black, fill=black] (-1.73,0) circle (2pt) node[anchor=south east] {$P_-$};
\filldraw[color=black, fill=black] (0.57,0) circle (2pt) node[anchor=south west] {$Q_-$};
\filldraw[color=black, fill=black] (-0.57,1) circle (2pt) node[anchor=south west] {$P_0$};
\filldraw[color=black, fill=black] (-1,1.27) circle (2pt) node[anchor=south east] {$P_{-1}$};
\filldraw[color=black, fill=black] (1,4.73) circle (2pt) node[anchor=south east] {$P_{1}$};
\draw[thick][postaction={decorate}] (-1.73,0) -- (-1,1.27);
\draw[thick][postaction={decorate}] (1,4.73) -- (-1,1.27);
\draw[thick][postaction={decorate}] (1,4.73) -- (3.4,8.9);
\draw[dashed] (1.73,0) -- (1.73,9);
\draw[thick][postaction={decorate}] (1.6,9) .. controls (1.6,6) ..(1,4.73);
\draw[thick][postaction={decorate}] (1.4,9) .. controls (0,2) ..(-5,9);
\draw[thick][postaction={decorate}] (1.68,9) .. controls (1.6,6) ..(3.2,8.9);
\draw (0.9,6) node[anchor=south west] {$S_{1}^+$};
\end{tikzpicture}
\end{center}

We firstly consider the extensions for which the exterior solution corresponds to the orbit $S_1^+:$

\begin{theorem}\label{solutions singular null boundary}
    Every interior Lorentzian Hawking--Page solution can be extended to a globally hyperbolic scale-invariant vacuum spacetime with an $SO(3)\times U(1)$ isometry that has a singular null boundary. The spacetime is globally covered by self-similar Bondi coordinates. The exterior solution is defined on $\mathcal{E}^{(4+1)}=\{(u,r,\theta,\varphi,\psi)|\ u<0,\ r\geq r_*(u)\}.$ 
    The null boundary $\mathcal{B}_{null}$ is reached along incoming null cones in the limit $u\rightarrow0-$ and represents a curvature singularity. We represent the Penrose diagram of the quotient of the spacetime by the $SO(3)\times U(1)$ action:
\end{theorem}
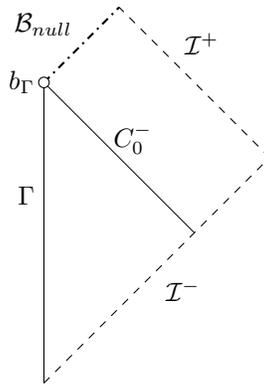
\begin{center}
\begin{tikzpicture}
\draw[dashed] (0,0) -- (3,3);
\draw (0,0) -- (0,4);
\draw[very thin] (0,4) -- (2,2);
\draw[dashed] (3,3) -- (1,5);
\draw[thick, dashdotted] (0,4) -- (1,5);
\filldraw[color=black, fill=white] (0,4) circle (2pt) node[anchor=east] {$b_{\Gamma}$};
\draw (0,2.5)  node[anchor=east] {$\Gamma$};
\draw (0.8,3.25)  node[anchor=west] {$C_0^-$};
\draw (1.5,1.5)  node[anchor=north west] {$\mathcal{I}^-$};
\draw (1.75,4.25)  node[anchor=south west] {$\mathcal{I}^+$};
\draw (0.5,4.5)  node[anchor=south east] {$\mathcal{B}_{null}$};
\end{tikzpicture}
\end{center}
\begin{proof}
    We extend the interior Lorentzian Hawking--Page solution by choosing the orbit corresponding to the exterior solution to be $S_1^+.$ The orbit ends at the critical point $P_1,$ so the exterior solution is defined on $\mathcal{E}^{(4+1)}=\{s_*\leq s<\infty\}.$ As $s\xrightarrow[]{}\infty,$ we have the limits:
    \[\theta\xrightarrow[]{}1,\ \zeta\xrightarrow[]{}-1,\ e^{2\lambda}\xrightarrow[]{}2,\ \mu\rightarrow\frac{1}{2}.\]
    
    We define another null coordinate $v$, which is constant along incoming null cones and increases along outgoing null cones, such that $v=u$ on $\Gamma$. The causality of $\mathcal{Q}^{(1+1)}$ is the same as that of $\mathbb{R}^{1+1}_{(u,v)}.$ Moreover, in this gauge we have $C_0^-=\{v=0\}.$
    
    We consider the null boundary $\mathcal{B}_{null}=\{v\geq u,\ u=0\}$. We follow the argument of \cite{nakedsingularities} to show this is indeed singular. Along the incoming null cone which passes through $(u_0,s_0)$ in the cone exterior region, we have:
    \[-\log(-u)=-\log(-u_0)+\int_{s_0}^s\alpha(s')ds'.\]  
    We rewrite this as follows:
    \[-\log(-u)-\frac{2s}{1-k}=-\log(-u_0)-\frac{2s_0}{1-k}+\int_{s_0}^s\alpha(s')-\frac{2}{1-k}ds'.\]
    A similar analysis of (\ref{system alpha theta}) near $P_1$ as in Appendix \ref{Appendix Expansion Subsection} shows that $\alpha=\frac{2}{1-k}+O(e^{-2s})$ as $s\xrightarrow[]{}\infty$. Thus, the right hand side converges to a finite limit as $s\xrightarrow[]{}\infty$. Using $s=\log\big(\frac{r}{-u}\big),$ we get that:
    \[r(u,v)\sim(-u)^{\frac{1+k}{2}}\text{ as }u\xrightarrow[]{}0.\]
    From $\zeta(u,v)\xrightarrow[]{}-1$ as $u\xrightarrow[]{}0$, we get:
    \[\phi(u,v)=-\log r(u,v)+O(1)\text{ as }u\xrightarrow[]{}0.\]
    As a result, we note that for any $v>v_*$:
    \[f(u,v)=\exp\bigg(-\frac{2}{\sqrt{3}}\phi\bigg)\sim r(u,v)^{\frac{2}{\sqrt{3}}}\text{ as }u\xrightarrow[]{}0.\]
    
    We compute the sectional curvature of the spheres according to Appendix \ref{Appendix sectional curvature}:
    \[K^{(4+1)}_{S^2(u,v)}=\frac{f}{r^2}\big(1-e^{-2\lambda}(k\theta+1)(k\zeta+1)\big)\sim (-u)^{-\frac{2}{3}}\text{ as }u\xrightarrow[]{}0,\]
    so $K^{(4+1)}_{S^2(u,v)}$ blows up along any incoming null cone originating in the cone exterior region. We conclude that the null boundary $\mathcal{B}_{null}$ is a curvature singularity.
\end{proof}

\begin{remark}
    For the corresponding $(3+1)$-dimensional solution $\big(\mathcal{Q}^{(1+1)}\times S^2,g,\phi\big),$ the sectional curvature of the spheres also blows up along any incoming null cone originating in the cone exterior region:
    \[K_{S^2(u,v)}=\frac{\mu}{r^2}\xrightarrow[]{}\infty\text{ as }u\xrightarrow[]{}0.\]
    Based on this, Christodoulou concluded in \cite{nakedsingularities} that for exterior solutions of the Einstein-Scalar field equations corresponding to the orbit $S_1^+,$ the null boundary $\mathcal{B}_{null}$ is a curvature singularity.
\end{remark}

\begin{remark}
    In Section \ref{asymptotically flat spacetimes}, we shall use the above solutions to construct extensions of the interior Lorentzian Hawking--Page solutions to asymptotically flat solutions with locally naked singularities.
\end{remark}
 
We now describe the maximal extension in the regular region of the exterior solutions whose orbits are different from $S_1^+:$

\begin{theorem}\label{solutions apparent horizon}
    For every interior Lorentzian Hawking--Page solution there exist extensions to exterior solutions which correspond to the orbits to the left and right of $S_1^+$. In self-similar Bondi gauge, each exterior solution in the regular region is defined on $\mathcal{E}^{(4+1)}\cap\mathcal{R}^{(4+1)}=\{(u,r,\theta,\varphi,\psi)|\ u<0,\ r_*(u)\leq r\leq r_{\mathcal{A}}(u)\},$ where $e^{s_{\mathcal{A}}}=-r_{\mathcal{A}}(u)/u.$ The set $\mathcal{A}^{(4+1)}=\{s=s_{\mathcal{A}}\}$ represents the apparent horizon.
\end{theorem}
\begin{proof}
    The first two statements follow by Proposition \ref{Christodoulou phase portrait analysis}. To prove the last statement, we repeat the argument of \cite{nakedsingularities}. Along the incoming null cone passing through $(u_0,s_0)$ in the cone exterior region, we have:
    \[\log(-u_0)-\log(-u)=\int_{s_0}^s\alpha(s')ds'.\]
    Since $\alpha\sqrt{s_{\mathcal{A}}-s}\xrightarrow[]{}\sigma_{\mathcal{A}}/\sqrt{2}$ as $s\xrightarrow[]{}s_{\mathcal{A}}$, we get that $\int_{s_0}^{s_{\mathcal{A}}}\alpha(s')ds'$ converges. This implies that the incoming null cone reaches $s=s_{\mathcal{A}}$ at some $u<0.$ We also get that $e^{2\lambda}\xrightarrow[]{}\infty$ as $s\xrightarrow[]{}s_{\mathcal{A}}$, so $\mu\xrightarrow[]{}1$ as $s\xrightarrow[]{}s_{\mathcal{A}}$. Thus, $\mathcal{A}^{(4+1)}=\{s=s_{\mathcal{A}}\}$ is the apparent horizon.
\end{proof}

\begin{remark}
    In \cite{nakedsingularities}, Christodoulou constructs examples of $(3+1)$-dimensional spacetimes with a naked singularity using orbits of $(\ref{system alpha theta})$ that end at $P_0$.  However, in our case $k^2=\frac{1}{3}$ implies that the line $\alpha=\frac{\theta}{k}+3$ separates the point $P_0$ from the orbits which have a vertical asymptote at $\theta=\frac{1}{k}$ in the upper half plane.
\end{remark}

\subsection{Extension Beyond the Regular Region}

In the previous section we extended the Lorentzian Hawking--Page solutions to the cone exterior region. The solutions obtained in Theorem \ref{solutions singular null boundary} are defined in the maximal region as globally hyperbolic solutions. However, the solutions from Theorem \ref{solutions apparent horizon} are only defined in the regular region so far, and can be extended further. 

It is well-known that solutions to the Einstein-Scalar field  system in spherical symmetry with $\mathcal{A}^{(3+1)}\neq 0$ can be extended beyond the apparent horizon into the trapped region $\mathcal{T}^{(3+1)},$ where all causal curves reach a spacelike singularity at the boundary $\mathcal{B},$ with $r=0.$ Here $\mathcal{B}$ represents the canonical boundary that we can attach to the quotient manifold $\mathcal{Q}^{(1+1)}.$ For the corresponding $(4+1)$-dimensional vacuum solution we have that either $\frac{r^2}{f}$ (the area radius of the spheres $S^2$) goes to zero at $\mathcal{B}$, or $f$ (the area radius of the spheres $S^1$) goes to zero at $\mathcal{B}$, or both. However, this argument is not sufficient to understand the nature of the boundary $\mathcal{B}$ for the $(4+1)$-dimensional vacuum solution.

In this section, we study the extension beyond the regular region $\mathcal{R}^{(4+1)}$ in detail, in order to describe the behavior of the solutions near the boundary $\mathcal{B}$. In one case, we obtain that $\mathcal{B}$ represents a spacelike curvature singularity. In the other case, we have that $\mathcal{B}$ represents a null Cauchy horizon of Taub--NUT type, beyond which the solutions have non-unique analytic extensions with closed timelike curves. Moreover, in this later case we obtain that the solutions in the cone exterior region are given explicitly, corresponding to a Wick rotation of the interior Lorentzian Hawking--Page solutions. This completes our goal to classify all the maximal extensions of Lorentzian Hawking--Page solutions in the class of $(4+1)$-dimensional scale-invariant vacuum spacetimes with an $SO(3)\times U(1)$ isometry. 

We remark that the analysis in \cite{nakedsingularities} is restricted to the regular region. Thus, this section is also of interest because it describes the behavior of the Einstein-Scalar field solutions of \cite{nakedsingularities} in the case $k^2=\frac{1}{3}$ in the trapped region. In particular, we compute explicitly the blow-up rate of the scalar field $\phi=c\log r+O(1)$ near $\mathcal{B},$ where $c=\pm\sqrt{3},\ \mp\frac{1}{\sqrt{3}}$.

\subsubsection{Bondi to double null coordinates}\label{Bondi to double null coordinates}
We consider the Einstein-Scalar field solution associated to our $(4+1)$-dimensional vacuum spacetime. Bondi coordinates cover the regular region $\mathcal{R}^{(3+1)}$, but they become singular at the apparent horizon $\mathcal{A}^{(3+1)}$. In order to extend the solution to the trapped region $\mathcal{T}^{(3+1)}$, we introduce suitable double null coordinates in the cone exterior region. Using the Kaluza--Klein reduction, we obtain double null coordinates on the $(4+1)$-dimensional spacetime and a system equivalent to the vacuum equations.

We firstly define an incoming null coordinate $V$ as follows: fix $u_0<0$, and define $V(u_0,\cdot)$ to be the affine parameter along $C_{u_0}^+$ with $lV=1,\ V(u_0,r_*(u_0))=0$; then extend $V$ off $C_{u_0}^+$ by asking that $V$ is constant along incoming null cones. In particular, we have $V=0$ on $C_0^-.$ 

We recall that in the self-similar Bondi coordinates introduced in Section \ref{introduction to bondi gauge}, we have:
\[g=-e^{2\nu}du^2-2e^{\nu+\lambda}dudr+r^2d\sigma_2^2,\ S=u\partial_u+r\partial_{r}.\] 
With respect to the double null coordinates $(u,V),$ this becomes:
\[g=-\Bar{\Omega}^2dudV+r^2d\sigma_2^2,\ S=u\partial_u+\frac{r\beta}{\partial_Vr}\partial_V,\]
where $\Bar{\Omega}^2=2e^{\nu+\lambda}\partial_Vr.$ We also compute:
\[\partial_ur=-\frac{1}{2}e^{\nu-\lambda},\ \partial_u=\frac{1}{2}e^{\nu}n,\ \partial_V=\big(e^{\lambda}\partial_Vr\big) l.\]

We notice that the self-similarity condition $(\mathcal{L}_Sg)_{VV}=0$ implies that $\frac{r\beta}{\partial_Vr}$ is independent of $u.$ Thus, we can renormalize the incoming null coordinate by defining for $V>0$:
\[v(u,V)=v(u_0,V)=V_{\mathcal{A}}\exp\bigg(-\int_V^{V_{\mathcal{A}}(u_0)}\frac{\partial_Vr}{r\beta}d\Tilde{V}\bigg)=V_{\mathcal{A}}\exp\bigg(-\int_{s(u_0,V)}^{s_{\mathcal{A}}}\frac{ds}{\beta}\bigg).\]

This is well defined, because of the expansion of $\beta$ near $\mathcal{A}^{(3+1)}$ from Proposition \ref{Christodoulou phase portrait analysis}. From Section \ref{Extending Beyond the Interior Region section} we also know that $\beta=\frac{2}{3}(s-s_*)+O(|s-s_*|^2)$, so $\lim_{V\rightarrow0+}v(V)=0$ and $C_0^-=\{v=0\}.$ The region $\text{int}\big(\mathcal{R}^{(3+1)}\cap\mathcal{E}^{(3+1)}\big)$ is covered by the double null coordinates $(u,v)$, such that:
\[g=-\Omega^2dudv+r^2d\sigma_2^2,\ \Bar{\Omega}^2=\Omega^2\frac{dv}{dV},\ S=u\partial_u+v\partial_v.\]

In general, we say that a solution to the Einstein-Scalar field equations is in \textit{self-similar double null gauge} if there exist double null coordinates such that $S=u\partial_u+v\partial_v.$ The gauge freedom left is:
\begin{equation}\label{gauge freedom double null}
    u\rightarrow c_1^2u,\ v\rightarrow c_2^2v.
\end{equation}
In order to extend the solution beyond $\mathcal{A}^{(3+1)}$, we need to write the Einstein-Scalar field system in self-similar double null gauge. In these coordinates, the self-similar condition becomes:
\[S\Omega=0,\ Sr=r,\ S\phi=-k.\]
As in Section \ref{introduction to bondi gauge}, we introduce the self-similar coordinate:
\[y=-\frac{v}{u}.\]
The change of coordinates is given by:
\begin{equation}\label{change of coordinates}
    \partial_v=-\frac{1}{u}\partial_y,\ 
    \partial_u=\frac{v}{u^2}\partial_y+\partial_u.
\end{equation}
Thus, in $(u,y)$ coordinates we have $S=u\partial_u.$ We introduce the notation:
\[\chi=\phi+k\log(-u),\ R=-\frac{r}{u}.\]
The self-similar condition is equivalent to:
\[\Omega=\Omega(y),\ R=R(y),\ \chi=\chi(y).\]
Using (\ref{change of coordinates}), we compute:
\[\beta=\frac{v\partial_vr}{r}=\frac{y\partial_yR}{R},\]
\[\zeta=\frac{\theta\beta+k}{\beta-1}.\]

According to the computations in Appendix \ref{Appendix Double Null}, the Einstein-Scalar field system is equivalent to:
\[\frac{1}{1-\mu}=\beta\theta^2-\frac{(\theta\beta+k)^2}{\beta-1}+1\]
\begin{equation}\label{double null systemmm}
    \begin{dcases}
    y\partial_y\theta=k\big(k\theta-1\big)+\beta\theta\big[(\theta+k)^2-(1+k^2)\big] \\
    y\partial_y\beta=\beta(1-k^2)-\beta^2\big[(\theta+k)^2+(1-k^2)\big]
    \end{dcases}
\end{equation}

We define the variable $y=e^{t}$ and we remark that in terms of $t$ we recover the system (\ref{system in t variable}). Similarly to $s$ in self-similar Bondi coordinates, we can also define the self-similar coordinate:
\[\tau(t)=\tau_*+\int_{-\infty}^{t}\beta d\tilde{t}.\]
This is well defined by the analysis of the negatively stable critical point $(\frac{1}{k},0)$ of (\ref{system in t variable}) in Section \ref{Extending Beyond the Interior Region section}. We remark that in the variable $\tau$ we recover the system (\ref{system theta beta}) for $(\theta,\beta)$. We also have that $C_0^-=\{\tau=\tau_*\}.$

\subsubsection{Extending beyond the apparent horizon}

We consider an exterior solution obtained in Theorem \ref{solutions apparent horizon}, which is defined on $\mathcal{E}^{(4+1)}\cap\mathcal{R}^{(4+1)}=\{t\in(-\infty,t_{\mathcal{A}})\}.$ In this section, we extend the solution beyond $\mathcal{A}^{(4+1)}$. According to the above section, it satisfies the system:
\begin{equation}\label{System in t variablee}
    \begin{dcases}
    \frac{d\theta}{dt}=k\big(k\theta-1\big)+\beta\theta\big[(\theta+k)^2-(1+k^2)\big] \\
    \frac{d\beta}{dt}=\beta(1-k^2)-\beta^2\big[(\theta+k)^2+(1-k^2)\big]
    \end{dcases}
\end{equation}

According to Proposition \ref{Christodoulou phase portrait analysis}, we have that there exists $\sigma_{\mathcal{A}}\neq0$ such that $\beta\xrightarrow[]{}0,\ \beta\theta\xrightarrow[]{}\eta_{\mathcal{A}}:=1/\sigma_{\mathcal{A}}$ as $t\rightarrow t_{\mathcal{A}}$. We introduce the function $\eta=\beta\theta$ and rewrite the system as:
\begin{equation}\label{eta beta system}
    \begin{dcases}
    \frac{d\eta}{dt}=-2\eta\beta-k\beta+\eta \\
    \frac{d\beta}{dt}=\beta\big(1-k^2-\beta\big)-2k\eta\beta-\eta^2
    \end{dcases}
\end{equation}
Since $(\eta_{\mathcal{A}},0)$ is not a critical point, we can solve uniquely for $t\geq t_{\mathcal{A}}$ and we obtain in a neighborhood of $\mathcal{A}^{(4+1)}:$
\[\eta(t)=\eta_{\mathcal{A}}+\eta_{\mathcal{A}}(t-t_{\mathcal{A}})+O(|t-t_{\mathcal{A}}^2|),\]
\[\beta(t)=-\eta_{\mathcal{A}}^2(t-t_{\mathcal{A}})+O(|t-t_{\mathcal{A}}^2|).\]
We remark that the following identities hold:
\[v\partial_v\phi=\eta,\]
\[\frac{v}{r}\cdot\frac{\partial_vr}{1-\mu}=\beta+\eta^2+\frac{\beta}{1-\beta}(\eta+k)^2.\]
It suffices to solve the system (\ref{eta beta system}), because the quantities on the left hand side of the above identities determine the solution to the Einstein vacuum equations, with boundary values given by $r$ and $\phi$ on $C_0^-.$ Therefore, we extended the vacuum solutions of Theorem \ref{solutions apparent horizon} for $t>t_{\mathcal{A}}.$

We consider the extended solutions in the $(\theta,\alpha)$ plane. They satisfy the system:
\begin{equation}\label{theta alpha dt system}
    \begin{dcases}
    \frac{d\theta}{dt}=k\big(k\theta-1\big)+\frac{\theta}{\alpha}\big[(\theta+k)^2-(1+k^2)\big] \\
    \frac{d\alpha}{dt}=(\theta+k)^2+(1-k^2)(1-\alpha)
    \end{dcases}
\end{equation}
The above analysis of (\ref{eta beta system}) shows that $\beta$ changes sign, whereas the sign of $\eta$ remains unchanged at $t=t_{\mathcal{A}}$. This allows us to describe the behavior of orbits corresponding to the extended solutions near $\mathcal{A}^{(4+1)}:$
\begin{itemize}
    \item The orbits to the left of $S_1^+$ blow up as $t\xrightarrow[]{}t_{\mathcal{A}}-$ with $\theta\xrightarrow[]{}-\infty,\ \alpha\xrightarrow[]{}\infty,$ and $\alpha/\theta\xrightarrow[]{}1/\eta_{\mathcal{A}}.$ They can be continued for $t>t_{\mathcal{A}}$ and they blow up as $t\xrightarrow[]{}t_{\mathcal{A}}+$ with $\theta\xrightarrow[]{}\infty,\ \alpha\xrightarrow[]{}-\infty,$ and $\alpha/\theta\xrightarrow[]{}1/\eta_{\mathcal{A}}.$
    \item The orbits to the right of $S_1^+$ blow up as $t\xrightarrow[]{}t_{\mathcal{A}}-$ with $\theta\xrightarrow[]{}\infty,\ \alpha\xrightarrow[]{}\infty,$ and $\alpha/\theta\xrightarrow[]{}1/\eta_{\mathcal{A}}.$ They can be continued for $t>t_{\mathcal{A}}$ and they blow up as $t\xrightarrow[]{}t_{\mathcal{A}}+$ with $\theta\xrightarrow[]{}-\infty,\ \alpha\xrightarrow[]{}-\infty,$ and $\alpha/\theta\xrightarrow[]{}1/\eta_{\mathcal{A}}.$
\end{itemize}

Since we are in the case $k^2=\frac{1}{3},$ we have the following result which is useful in the phase portrait analysis:

\begin{lemma}\label{side of line orbit}
    In the $(\theta,\alpha)$ plane, the line $\alpha=\frac{\theta}{k}+3$ contains the critical points $P_{-},\ P_{-1},\ P_1,$ and is an exact solution to (\ref{theta alpha dt system}) for $\theta<-1,\ -1<\theta<1,\ 1<\theta.$ The orbits corresponding to the extensions of solutions in Theorem \ref{solutions apparent horizon} satisfy:
    \[\theta<k\alpha-\sqrt{3}\text{ for }t<t_{\mathcal{A}}\]
    \[\theta>k\alpha-\sqrt{3}\text{ for }t>t_{\mathcal{A}}\]
\end{lemma}
\begin{proof}
    The first statement follows by a direct computation. Since our orbits have a vertical asymptote at $\theta=\frac{1}{k}$ in the upper half plane as $t\xrightarrow[]{}-\infty$, we get $\theta<k\alpha-\sqrt{3}\text{ for }t<t_{\mathcal{A}}$. In the $(\eta,\beta)$ plane the exact orbit is given by $\eta=k-\beta\sqrt{3}.$ For $t<t_{\mathcal{A}}$ we know that $\eta<k-\beta\sqrt{3},$ which implies this holds for $t>t_{\mathcal{A}}$ as well. Using that $\alpha<0$ for $t>t_{\mathcal{A}}$, we conclude that $\theta>k\alpha-\sqrt{3}\text{ for }t>t_{\mathcal{A}}.$
\end{proof}

\subsubsection{The trapped region}
In this section, we study of the extensions of the solutions from Theorem \ref{solutions apparent horizon} in the trapped region $\mathcal{T}^{(4+1)}$. We begin with the following simple observations:
\begin{lemma}
    The region $\{t>t_{\mathcal{A}}\}$ is the trapped region $\mathcal{T}^{(4+1)}$.
\end{lemma}
\begin{proof}
    For $t>t_{\mathcal{A}}$ we have $\beta>0$ so $\partial_vr<0.$ We also have:
    \[\frac{1}{1-\mu}=1+\beta\theta^2+\frac{(\theta\beta+k)^2}{1-\beta}>1.\]
\end{proof}

\begin{lemma}\label{orbits spacelike singularity}
    The solutions reach $\alpha=0$ at some $t_{\mathcal{B}}<\infty.$ We define $\mathcal{B}=\{t=t_{\mathcal{B}}\}.$ Then we have $r=0$ on $\mathcal{B}.$
\end{lemma}
\begin{proof}
    The system (\ref{theta alpha dt system}) implies that $\frac{d\alpha}{dt}>1-k^2.$ Since $\alpha<0$ for $t>t_{\mathcal{A}},$ there exists $t_{\mathcal{A}}<t_{\mathcal{B}}<\infty$ such that the solution blows up. If we assume that the orbit has the horizontal asymptote $\alpha=A\leq0$ as $\theta\xrightarrow[]{}\infty,$ we get the contradiction $\frac{d\theta}{dt}<-\frac{1}{1-A}\cdot\theta^3$ for all $\theta$ large enough. Therefore, we obtain that
    $\alpha=0$ at $\{t=t_{\mathcal{B}}\}.$ Since $\alpha=\frac{r}{v\partial_vr}$, we obtain that each point of $\mathcal{B}$ is a first singularity away from the center for the corresponding Einstein-Scalar field solution. Using the extension criterion of \cite{kommemi}, we conclude that $r=0$ on $\mathcal{B}.$
\end{proof}

This result also confirms the fact that the corresponding solutions to the Einstein-Scalar field system in spherical symmetry can be extended into the trapped region up to a spacelike singular boundary $\mathcal{B}$, with $r = 0$. As argued above, we cannot obtain the same conclusion directly for the $(4+1)$-dimensional vacuum spacetime. We further study the behavior of the vacuum solution near $\mathcal{B}$.

We recall that we defined the variable $\tau$ in the region $\mathcal{E}^{(4+1)}\cap\mathcal{R}^{(4+1)},$ which increases along future null cones. In the region $\mathcal{T}^{(4+1)}$ we define:
\[\tau(t)=\tau_{\mathcal{A}}-\int_{t_{\mathcal{A}}}^{t}\frac{1}{\alpha} d\tilde{t}.\]
We remark that $\alpha<0$ in $\mathcal{T}^{(4+1)}$ implies that $\tau$ increases with respect to $t$. In the trapped region, the system (\ref{theta alpha dt system}) is equivalent to:
\begin{equation}\label{system alpha theta trapped}
    \begin{dcases}
    \frac{d\theta}{d\tau}=-k\alpha\big(k\theta-1\big)-\theta\big[(\theta+k)^2-(1+k^2)\big] \\
    \frac{d\alpha}{d\tau}=-\alpha\big[(\theta+k)^2+(1-k^2)(1-\alpha)\big]
    \end{dcases}
\end{equation}

We study the phase portrait of (\ref{system alpha theta trapped}), in order to understand the solutions in the region $\mathcal{T}^{(4+1)}$. The critical points with $\alpha\leq0$ are:
\begin{itemize}
    \item Critical points stable to the future: $P_-=(-\sqrt{3},0),\ Q_-=(k,0).$
    \item Saddle critical points: $(0,0).$
\end{itemize}
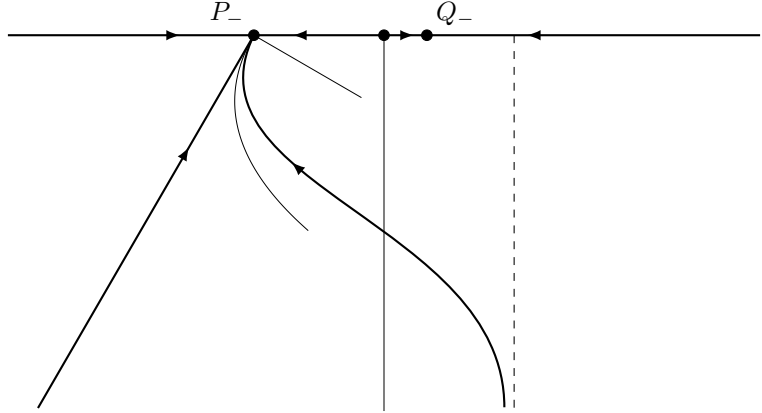
\begin{center}
\begin{tikzpicture}[decoration={markings, mark=at position 0.7 with {\arrow{latex}}}] 
\draw[dashed] (1.73,0) -- (1.73,-5);
\draw[thick][postaction={decorate}] (1.6,-4.95) .. controls (1.6,-2.5) and (-2.7,-2) ..(-1.73,0);
\draw[thick][postaction={decorate}] (-5,0) -- (-1.73,0);
\draw[thick][postaction={decorate}] (0,0) -- (-1.73,0);
\draw[thick][postaction={decorate}] (0,0) -- (0.57,0);
\draw[thick][postaction={decorate}] (5,0) -- (0.57,0);
\draw[very thin] (0,-5) -- (0,0);
\filldraw[color=black, fill=black] (0,0) circle (2pt);
\filldraw[color=black, fill=black] (-1.73,0) circle (2pt) node[anchor=south east] {$P_-$};
\filldraw[color=black, fill=black] (0.57,0) circle (2pt) node[anchor=south west] {$Q_-$};
\draw[thick][postaction={decorate}] (-4.6,-4.96) -- (-1.73,0);
\draw[very thin] (-1.73,0) -- (-0.3,-0.83);
\draw[very thin]   plot[smooth,domain=-2.6:0] ({-1.73+(0.58*\x)+(0.33*(\x)^2)},\x);
\end{tikzpicture}
\end{center}
Moreover, we note that $\{\alpha=0\}$ is an orbit of (\ref{system alpha theta trapped}) on $\theta<-\sqrt{3},\ -\sqrt{3}<\theta<k,\ \theta>k$. Thus, all orbits starting at $\alpha\xrightarrow[]{}-\infty$ remain in the lower half plane and end at $P_-,\ Q_-$ or $(0,0).$ We also recall that the line $\alpha=\frac{\theta}{k}+3$ is an orbit of the system. Finally, the orbits corresponding to interior Lorentzian Hawking--Page solutions are also orbits of (\ref{system alpha theta trapped}), but with reverse orientation.

\begin{lemma}
    The orbits corresponding to solutions in $\mathcal{T}^{(4+1)}$ do not reach the critical point $(0,0).$
\end{lemma}
\begin{proof}
    Since $(0,0)$ is a saddle point, there could be at most one such orbit with $\alpha<0$. Along this orbit we have:
    \[\frac{1}{1-\mu}=1+k^2+\frac{(\theta+k)^2}{\alpha-1}\xrightarrow[]{}1,\]
    which contradicts the fact that $\mu>1$ in $\mathcal{T}^{(4+1)}$.
\end{proof}

\begin{lemma}
The orbits corresponding to solutions in $\mathcal{T}^{(4+1)}$ that blow up as $\tau\xrightarrow[]{}\tau_{\mathcal{A}}+$ with $\theta\xrightarrow[]{}\infty,\ \alpha\xrightarrow[]{}-\infty,$ and $\alpha/\theta\xrightarrow[]{}1/\eta_{\mathcal{A}}$ reach $Q_-$ as $\tau\rightarrow\infty.$  
\end{lemma}
\begin{proof}
    By the study of interior Lorentzian Hawking--Page solutions, we know there exist orbits of (\ref{system alpha theta trapped}) with the vertical asymptote $\theta=k$ in the lower half plane that end at $P_-.$ Fix one such orbit, and say it intersects $\{\theta=0\}$ at $\alpha=\alpha_0.$ Since $\mu<1$ at $(0,\alpha_0)$, we see that:
    \[\frac{1}{1-\mu}(0,\alpha)=1+k^2+\frac{k^2}{\alpha-1}\]
    implies $\mu(0,\alpha)<1$ for all $\alpha\in(\alpha_0,0).$ We notice that the orbits which blow up with $\theta\xrightarrow[]{}\infty,\ \alpha\xrightarrow[]{}-\infty$ lie to the right of our fixed interior Lorentzian Hawking--Page orbit. Since they correspond to solutions of the Einstein-Scalar field equations in the trapped region, they cannot intersect $\theta=0.$ We conclude that they must end at $Q_-$.
\end{proof}

This result allows us to complete the description of the maximal extension of solutions corresponding to orbits to the left of $S_1^+$ in Theorem \ref{solutions apparent horizon}:
\begin{theorem}\label{spacelike singularity}
    Consider a Lorentzian Hawking--Page solution, which in the exterior regular region $\mathcal{E}^{(4+1)}\cap\mathcal{R}^{(4+1)}$ corresponds to an orbit to the left of $S_1^+$, the upper branch of the stable manifold of $P_1$ in the $(\theta,\alpha)$ plane. This can be extended uniquely beyond the apparent horizon $\mathcal{A}^{(4+1)}$ into the trapped region $\mathcal{T}^{(4+1)}$ and it has a spacelike curvature singularity at the boundary $\mathcal{B}$ with $r^2/f\rightarrow0$ and $f\rightarrow\infty$. We represent the Penrose diagram of the quotient of the spacetime by the $SO(3)\times U(1)$ action:
\end{theorem}
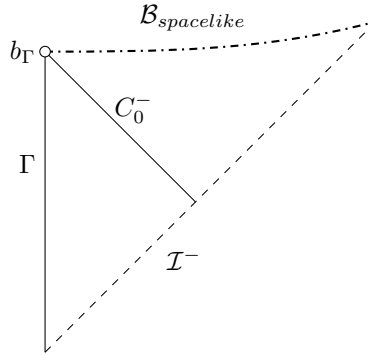
\begin{center}
\begin{tikzpicture}
\draw[dashed] (0,0) -- (4.4,4.4);
\draw (0,0) -- (0,4);
\draw[very thin] (0,4) -- (2,2);
\draw[thick, dashdotted] (0,4) .. controls (1.5,4) and (2.9,4) .. (4.4,4.4);
\filldraw[color=black, fill=white] (0,4) circle (2pt) node[anchor=east] {$b_{\Gamma}$};
\draw (0,2.5)  node[anchor=east] {$\Gamma$};
\draw (0.8,3.25)  node[anchor=west] {$C_0^-$};
\draw (1.5,1.5)  node[anchor=north west] {$\mathcal{I}^-$};
\draw (2.8,4.15)  node[anchor=south east] {$\mathcal{B}_{spacelike}$};
\end{tikzpicture}
\end{center}
\begin{proof}
    The above lemmas show that the region $\mathcal{T}^{(4+1)}$ is given by $\{t_{\mathcal{A}}<t<t_{\mathcal{B}}\}$, and we have that $r=0$ at the boundary $\mathcal{B}=\{t=t_{\mathcal{B}}\}.$ Since the orbit of $(\ref{system alpha theta trapped})$ corresponding to the solution in the region $\mathcal{T}^{(4+1)}$ converges to $Q_-$, we have $\lim_{t\xrightarrow[]{}t_{\mathcal{B}}}\theta=\lim_{t\xrightarrow[]{}t_{\mathcal{B}}}\zeta=k.$ This gives $\phi=k\log r+O(1),$ so $f\sim r^{-2/3}$ as $t\xrightarrow[]{}t_{\mathcal{B}}.$ We also have that $\mu\xrightarrow[]{}\infty$ as $t\xrightarrow[]{}t_{\mathcal{B}}.$  We compute the sectional curvature of the spheres:
    \[K^{(4+1)}_{S^2(u,v)}=\frac{f}{r^2}\big(1-(1-\mu)(k\theta+1)(k\zeta+1)\big)\xrightarrow[]{}\infty\text{ as }t\xrightarrow[]{}t_{\mathcal{B}}.\]
    Thus, the curvature blows up at the spacelike singularity $\mathcal{B}$. Also, we notice that the area radius of the spheres $S^2$ satisfies $r^2/f\rightarrow0$, and the area radius of the spheres $S^1$ satisfies $f\rightarrow\infty$.
\end{proof}

\begin{remark}
    We note that the above proof also shows that in the case of the Einstein-Scalar field solutions of \cite{nakedsingularities} with $k^2=\frac{1}{3}$, which correspond to orbits to the left of $S_1^+$, the boundary $\mathcal{B}$ represents a spacelike curvature singularity with blow-up rate of the scalar field $\phi=k\log r+O(1).$
\end{remark}

The only extensions of Lorentzian Hawking--Page solutions that we are still to consider are the ones corresponding to orbits to the right of $S_1^+$. We study the behavior of such orbits in the trapped region in detail:

\begin{center}
\begin{tikzpicture}
\draw[very thin] (-5,0) -- (5,0);
\draw[very thin] (0,-3) -- (0,0);
\filldraw[color=black, fill=black] (0,0) circle (2pt);
\filldraw[color=black, fill=black] (-1.73,0) circle (2pt) node[anchor=south east] {$P_-$};
\filldraw[color=black, fill=black] (0.57,0) circle (2pt) node[anchor=south west] {$Q_-$};
\draw[thick] (-3.5,-3.1) -- (-1.73,0);
\draw[very thin] (-1.73,0) -- (-0.3,-0.83);
\draw[very thin]   plot[smooth,domain=-2.6:0] ({-1.73+(0.58*\x)+(0.33*(\x)^2)},\x);
\draw[thick]   plot[smooth,domain=-1.5:0] ({-1.73+(0.58*\x)+((-0.26+0.33)*(\x)^2)},\x);
\draw[thick]   plot[smooth,domain=-1.5:0] ({-1.73+(0.58*\x)+((-0.2+0.33)*(\x)^2)},\x);
\draw[thick]   plot[smooth,domain=-1.5:0]({-1.73+(0.58*\x)+((-0.13+0.33)*(\x)^2)},\x);
\draw[thick]   plot[smooth,domain=-1.5:0] ({-1.73+(0.58*\x)+((-0.08+0.33)*(\x)^2)},\x);
\draw (-2.2,-0.8)  node[anchor=south east] {$X<-\frac{1}{3}$};
\draw (-1.1,-3.1)  node[anchor=south east] {$-\frac{1}{3}<X<0$};
\draw (-0.7,-1.2)  node[anchor=south east] {$X>0$};
\end{tikzpicture}
\end{center}

\begin{lemma}\label{expansion at B lemma}
The orbits corresponding to solutions in $\mathcal{T}^{(4+1)}$ that blow up as $\tau\xrightarrow[]{}\tau_{\mathcal{A}}+$ with $\theta\xrightarrow[]{}-\infty,\ \alpha\xrightarrow[]{}-\infty,$ and $\alpha/\theta\xrightarrow[]{}1/\eta_{\mathcal{A}}$ reach $P_-$ as $\tau\rightarrow\infty.$ For each orbit, there exists $X\in\big(-\frac{1}{3},0\big)$ such that near $P_-$ we have the expansions:
\[\theta=-\sqrt{3}+\frac{\alpha}{\sqrt{3}}+\bigg(X+\frac{1}{3}\bigg)\alpha^2+O(|\alpha|^3),\]
\[\frac{1}{1-\mu}=X\alpha^2+O(|\alpha|^3).\]
Moreover, for every $X\in\big(-\frac{1}{3},0\big)$ there exists at most one such orbit.
\end{lemma}
\begin{proof}
    Because of the interior Lorentzian Hawking--Page solutions, there exists an orbit with vertical asymptote $\theta=k$ in the lower half plane that ends at $P_-.$ As a result, the orbits which blow up with $\theta\xrightarrow[]{}-\infty,$ $\alpha\xrightarrow[]{}-\infty$ lie to the left of our fixed Lorentzian Hawking--Page orbit, so they must reach $P_-$ as $\tau\rightarrow\infty.$
    
    By the arguments in Appendix \ref{Appendix Expansion Subsection} and Section \ref{expansion near center section}, we obtain that any orbits of (\ref{system alpha theta trapped}) that end at $P_-$ satisfy the above expansions for some unique $X$. Using Lemma \ref{side of line orbit}, we get that the orbits are to the right of the line $\alpha=\frac{\theta}{k}+3$, so $X>-\frac{1}{3}$. Since the orbits correspond to solutions in the region $\{\mu>1\}$, we also get that $X\leq0.$ 
    
    Suppose now that there exists a solution with $X=0,$ and denote its orbit by $\gamma_0.$ For $\tau>\tau_{\mathcal{A}},$ we have that the above orbits satisfy $k\alpha-3k<\theta\leq \theta_{\gamma_0}$. We consider the system (\ref{eta beta system}) in the $(\eta,\beta)$ plane, and get that near $t=t_{\mathcal{A}}$ the above orbits satisfy $k\alpha-3k\beta<\eta\leq \eta_{\gamma_0}.$ Back to the system (\ref{system alpha theta}) in the $(\theta,\alpha)$ plane, we get that for $\tau<\tau_{\mathcal{A}}$ and $\theta$ large enough the above orbits satisfy $\theta_{\gamma_0}\leq\theta<k\alpha-3k.$ However, for the system (\ref{system in t variable}) in the $(\theta,\beta)$ plane near $\big(\frac{1}{k},0\big)$, the set of orbits corresponding to exterior solutions is open, because we could select any parameter $a_1^{(e)}$ when extending the cone interior solution. As a result, the set of orbits to the right of $S_1^+$ of (\ref{system alpha theta}) in the $(\theta,\alpha)$ plane corresponding to exterior solutions is open, contradicting the properties of $\gamma_0.$ We conclude that $X\neq0.$
\end{proof}

We now describe the extension of solutions corresponding to orbits to the right of $S_1^+$ in Theorem \ref{solutions apparent horizon}:

\begin{theorem}\label{regular spacelike boundary}
    Consider a Lorentzian Hawking--Page solution which in the exterior regular region $\mathcal{E}^{(4+1)}\cap\mathcal{R}^{(4+1)}$ corresponds to an orbit to the right of $S_1^+$, the upper branch of the stable manifold of $P_1$ in the $(\theta,\alpha)$ plane. This can be extended uniquely beyond the apparent horizon $\mathcal{A}^{(4+1)}$ into the trapped region $\mathcal{T}^{(4+1)}$ and it has a boundary $\mathcal{B}$ with $r^2/f\rightarrow const.\neq 0$ and $f\rightarrow0.$ We represent the Penrose diagram of the quotient of the spacetime by the $SO(3)\times U(1)$ action:
\end{theorem}
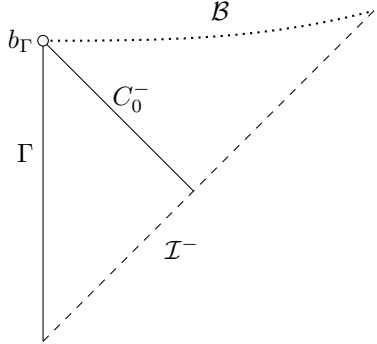
\begin{center}
\begin{tikzpicture}
\draw[dashed] (0,0) -- (4.4,4.4);
\draw (0,0) -- (0,4);
\draw[very thin] (0,4) -- (2,2);
\draw[thick, dotted] (0,4) .. controls (1.5,4) and (2.9,4) .. (4.4,4.4);
\filldraw[color=black, fill=white] (0,4) circle (2pt) node[anchor=east] {$b_{\Gamma}$};
\draw (0,2.5)  node[anchor=east] {$\Gamma$};
\draw (0.8,3.25)  node[anchor=west] {$C_0^-$};
\draw (1.5,1.5)  node[anchor=north west] {$\mathcal{I}^-$};
\draw (2.6,4.15)  node[anchor=south east] {$\mathcal{B}$};
\end{tikzpicture}
\end{center}
\begin{proof}
    Similarly to the proof of Theorem \ref{spacelike singularity}, the region $\mathcal{T}^{(4+1)}$ is given by $\{t_{\mathcal{A}}<t<t_{\mathcal{B}}\}$, and we have that $r=0$ at the boundary $\mathcal{B}=\{t=t_{\mathcal{B}}\}.$ The orbit of $(\ref{system alpha theta trapped})$ corresponding to the solution in the region $\mathcal{T}^{(4+1)}$ converges to $P_-$, so we get $\lim_{t\xrightarrow[]{}t_{\mathcal{B}}}\theta=\lim_{t\xrightarrow[]{}t_{\mathcal{B}}}\zeta=-1/k.$ This gives $\phi=-\frac{1}{k}\log r+O(1),$ so $f\sim r^{2}$ as $t\xrightarrow[]{}t_{\mathcal{B}}.$ The area radius of the spheres $S^1$ satisfies $f\rightarrow0$. Also, we have that $r^2/f$, the area radius of the spheres $S^2$, converges to a nonzero limit.
\end{proof}

\begin{remark}
    Once again, the above proof also shows that in the case of the Einstein-Scalar field solutions of \cite{nakedsingularities} with $k^2=\frac{1}{3}$, which correspond to orbits to the right of $S_1^+$, the boundary $\mathcal{B}$ represents a spacelike curvature singularity with blow-up rate of the scalar field $\phi=-\frac{1}{k}\log r+O(1).$ We also note that for $k^2=\frac{1}{3}$, \cite{nakedsingularities} provides an explicit analytic solution of the Einstein-Scalar field equations, which is homogeneous in space and satisfies $\phi=-\frac{1}{k}\log r+O(1).$ Our computation shows that the orbit of $(\ref{system alpha theta trapped})$ corresponding to this solution in the trapped region converges to $P_-$.
\end{remark}

We can compute that for the above solutions the sectional curvatures $K^{(4+1)}_{S^2(u,v)}$ and $K^{(4+1)}_{\mathcal{Q}(u,v)}$ are finite near $\mathcal{B}$. Unlike Theorem \ref{spacelike singularity}, we do not obtain directly that the solutions have a curvature singularity. We also point out that while the quotient by the $SO(3)\times U(1)$ action represented in the Penrose diagram above has a spacelike boundary, we cannot conclude that this corresponds to a spacelike boundary for the $(4+1)$-dimensional vacuum spacetime. We address these questions in the next section.

\subsubsection{Explicit solutions in the cone exterior region}\label{regular spacelike boundary section}

In this section, we find explicit solutions in the cone exterior region that are analogous to the Lorentzian Hawking--Page solutions. These represent a rigorous instance of the Wick rotation of cone interior solutions considered in \cite{Andersoncone} and \cite{Anderson}. In Proposition \ref{Taub-NUT extension} we extend these solutions non-uniquely beyond a null Cauchy horizon of Taub--NUT type to a region with closed timelike curves. In Theorem \ref{regularity at BB} we prove a similar result to Proposition \ref{interior all solutions}, by showing that every solution of Theorem \ref{regular spacelike boundary} is given by one of the explicit solutions in the cone exterior region. In particular, this shows that for the extensions of Lorentzian Hawking–Page solutions via Wick rotation we obtain a null Cauchy horizon of Taub--NUT type.

We briefly prove a series of results similar to the ones in Section \ref{first big section}:

\begin{proposition}
    Let $(M,g^{(3+1)})$ be a $(3+1)$-dimensional Lorentzian manifold. Then the $(4+1)$-dimensional Lorentzian manifold $\big((0,\infty)\times M, g^{(4+1)}\big),$ where:
    \[g^{(4+1)}=dz^2+z^2g^{(3+1)},\]
    is a vacuum solution if and only if $Ric(g^{(3+1)})=3g^{(3+1)}.$ The metric $g^{(4+1)}$ is scale-invariant with scaling vector field $S=z\partial_{z}$. The Kretschmann scalar is given by:
    \[{}^{(4+1)}K={}^{(4+1)}R_{\alpha\beta\gamma\delta}{}^{(4+1)}R^{\alpha\beta\gamma\delta}=z^{-4}\cdot{}^{(3+1)}R_{\alpha\beta\gamma\delta}{}^{(3+1)}R^{\alpha\beta\gamma\delta}=z^{-4}\cdot{}^{(3+1)}K.\]
\end{proposition}
\begin{proof} We compute the following:
    \[{}^{(4+1)}\Gamma_{ij}^k={}^{(3+1)}\Gamma_{ij}^k,\ {}^{(4+1)}\Gamma_{ij}^z=-zg_{ij},\ {}^{(4+1)}\Gamma_{zi}^i=\frac{1}{z},\]
    \[{}^{(4+1)}R_{ijkl}=z^2\cdot\bigg[{}^{(3+1)}R_{ijkl}-{}^{(3+1)}g_{ik}{}^{(3+1)}g_{jl}+{}^{(3+1)}g_{il}{}^{(3+1)}g_{jk}\bigg],\ {}^{(4+1)}R_{zijk}={}^{(4+1)}R_{zizj}=0,\]
    \[{}^{(4+1)}R_{ij}={}^{(3+1)}R_{ij}-3g_{ij},\ {}^{(4+1)}R_{zi}={}^{(4+1)}R_{zz}=0.\]
\end{proof}

We introduce an analogue of the function $V$ in Section \ref{Hawking--Page solutions section}. For any $M>-\frac{1}{3\sqrt{3}}$, we consider the function:
\[W(t)=t^2-1-\frac{2M}{t}.\]
Let $t_+$ be the largest root of $W$. Since $M>-\frac{1}{3\sqrt{3}},$ we have that $W'(t_+)>0,$ so $t_+$ is a simple root. This implies $t_+>-3M$, and we always have $t_+>0$. For any parameter $a>0,$ we define the following metrics on $(t_+,\infty)\times S^2\times S^1$:
\[g^{(3+1)}_{M,a}=-W^{-1}dt^2+t^2d\sigma_2^2+a^2Wd\sigma_1^2.\]
\begin{proposition}
    The Lorentzian metrics $g^{(3+1)}_{M,a}$ have an $SO(3)\times U(1)$ isometry with a free action. They satisfy the Einstein equations $Ric(g_{M,a}^{(3+1)})=3g_{M,a}^{(3+1)}.$ Moreover, the Kretschmann scalar is finite.
\end{proposition}
\begin{proof}
Let $g^{(3+1)}$ be a Lorentzian metric with an $SO(3)\times U(1)$ isometry:
\[g^{(3+1)}=-H^2(t)dt^2+t^2d\sigma_2^2+F^2(t)d\sigma_1^2.\]
The Einstein equations $Ric(g^{(3+1)})=3g^{(3+1)}$ are equivalent to:
\[\begin{dcases}
    \frac{F'}{F}+\frac{H'}{H}=0, \\
    \frac{2}{t}\frac{H'}{H}=\frac{H^2+1}{t^2}-3H^2.
\end{dcases}\]
We notice that for $W=H^{-2},$ the second equation is satisfied. By the first equation we get $F^2=a^2W,$ for some constant $a>0.$ We now compute the Kretschmann scalar. The Christoffel symbols are given by:
\[\Gamma_{tt}^t=-\frac{W'}{2W},\ \Gamma_{\psi\psi}^t=-\frac{a^2}{2}WW',\ \Gamma_{\theta\theta}^t=tW,\ \Gamma_{\varphi\varphi}^t=tW\sin^2\theta,\]
\[\Gamma_{t\psi}^{\psi}=\frac{W'}{2W},\ \Gamma_{t\theta}^{\theta}=\Gamma_{t\varphi}^{\varphi}=-\frac{1}{t},\ \Gamma_{\varphi\varphi}^{\theta}=-\sin\theta\cos\theta,\ \Gamma_{\varphi\theta}^{\varphi}=\frac{\cos\theta}{\sin\theta}.\]
We consider the orthonormal frame $e_1=\sqrt{W}\partial_t,\ e_2=(a\sqrt{W})^{-1}\partial_{\psi},\ e_3=t^{-1}\partial_{\theta},\ e_4=(t\sin\theta)^{-1}\partial_{\varphi}.$ With respect to this frame, the nonzero curvature components are:
\[R_{1212}=-\frac{1}{2}W'',\ R_{1313}=R_{1414}=R_{2323}=R_{2424}=\frac{W'}{2t},\ R_{3434}=\frac{1+W}{t^2}.\]
In particular, we obtain that the Kretschmann scalar is finite.
\end{proof}

Based on the previous results, we can construct the following spacetimes, which are the analogues of the Lorentzian Hawking--Page solutions:

\begin{definition}
    For every $M>-\frac{1}{3\sqrt{3}},\ a>0$ we define the Lorentzian metrics on $\mathcal{M}^{(4+1)}_M=(t_+,\infty)\times(0,\infty)\times S^2\times S^1$:
    \[g^{(4+1)}_{M,a}=-z^2W^{-1}dt^2+dz^2+z^2t^2d\sigma_2^2+a^2z^2Wd\sigma_1^2.\]
    These are scale-invariant vacuum spacetimes with an $SO(3)\times U(1)$ isometry and scaling vector field $S=z\partial_{z}$. The time orientation is given by $-\partial_{t},$ and the spacetimes are future geodesically incomplete. Moreover, the Kretschmann scalar is finite away from $z=0.$
\end{definition}

\begin{remark}
    The metrics $g^{(4+1)}_{M,a}$ defined above can be obtained by a Wick rotation of the Lorentzian Hawking--Page solutions. Formally, we take $\tau\rightarrow iz,\ \rho\rightarrow it,\ m\rightarrow -iM,\ b\rightarrow 2\pi a$, and we have $g_{m}^{(4+1)}\rightarrow g^{(4+1)}_{M,a}.$ Since the curvature of the Lorentzian Hawking--Page solutions is finite away from $\tau=0$, this transformation gives another proof that the above metrics have finite curvature away from $z=0.$
\end{remark}

In the following, we prove that the Lorentzian manifolds $\big(\mathcal{M}^{(4+1)}_M,\ g^{(4+1)}_{M,a}\big)$ can be extended beyond $\{t=t_+\},$ which represents a null Cauchy horizon of Taub--NUT type. The toy problem for this extension is the $(1+1)$-dimensional Lorentzian metric $-t^{-1}dt^2+td\psi^2$ on $(0,\infty)\times S^1$ considered by Misner in \cite{Misner}. Our discussion follows the argument in \cite{HawkingEllis}, which presents both Misner's example and the Taub--NUT space.

We introduce the tortoise coordinate:
\[t^*=\frac{1}{a}\int_{t_+}^t\frac{dt}{W},\]
so that $dt^*=(aW)^{-1}dt.$ We now define $\upsilon=\psi-t^*$ and remark that $(t,z,\theta,\varphi,\upsilon)$ are coordinates on $(t_+,\infty)\times(0,\infty)\times S^2\times S^1.$ The metric is given by:
\[g^{(4+1)}_{M,a}=2az^2dtd\upsilon+dz^2+z^2t^2d\sigma_2^2+a^2z^2Wd\upsilon^2.\]
This can be extended as a vacuum Lorentzian metric to ${}^{(1)}\mathcal{M}^{(4+1)}_M=(t_+-\epsilon,\infty)\times(0,\infty)\times S^2\times S^1.$ We notice that for any $t<t_+,$ there are closed timelike curves $\gamma(\upsilon)=(t_0,z_0,\theta_0,\varphi_0,\upsilon),$ so the extension is not globally hyperbolic. As a result, $\mathcal{CH}=\{t_+\}\times(0,\infty)\times S^2\times S^1$ represents a null Cauchy horizon of $\big(\mathcal{M}^{(4+1)}_M,\ g^{(4+1)}_{M,a}\big)$. Since the extension of the Taub metric into the NUT region has a similar behavior, we refer to $\mathcal{CH}$ as a null Cauchy horizon of Taub--NUT type. We represent the causality of the extended spacetime in the following diagram:

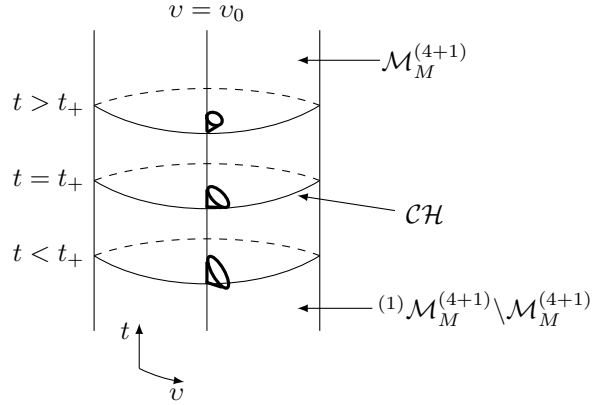
\begin{center}
\begin{tikzpicture}
    \draw (0,0) -- (0,4);
    \draw (3,0) -- (3,4);
    \draw[very thin] (1.5,0) -- (1.5,4);
    \draw (1.5,4)  node[anchor=south] {$\upsilon=\upsilon_0$};
    \draw (0,2)  node[anchor=east] {$t=t_+$};
    \draw (4.4,1.55)  node {$\mathcal{CH}$};
    \draw [-latex](4,1.6) -- (2.7,1.8);
    \draw (0,2) .. controls (0.8,1.5) and (2.2,1.5) .. (3,2);
    \draw[dashed] (0,2) .. controls (0.8,2.3) and (2.2,2.3) .. (3,2);
    \draw[very thick, shift={(1.65,1.78)}, rotate=45] (0,0) ellipse (0.08 and 0.18);
    \draw[very thick] (1.5,1.65) -- (1.5,1.88);
    \draw[very thick] (1.48,1.65) -- (1.7,1.64);
    \draw (0,3)  node[anchor=east] {$t>t_+$};
    \draw (4.4,3.6)  node {$\mathcal{M}^{(4+1)}_M$};
    \draw [-latex](3.8,3.6) -- (2.7,3.6);
    \draw (0,3) .. controls (0.8,2.5) and (2.2,2.5) .. (3,3);
    \draw[dashed] (0,3) .. controls (0.8,3.3) and (2.2,3.3) .. (3,3);
    \draw[very thick, shift={(1.6,2.82)}, rotate=60] (0,0) ellipse (0.08 and 0.11);
    \draw[very thick] (1.5,2.65) -- (1.5,2.88);
    \draw[very thick] (1.5,2.64) -- (1.65,2.73);
    \draw (0,1)  node[anchor=east] {$t<t_+$};
    \draw (5.2,0.3)  node {${}^{(1)}\mathcal{M}^{(4+1)}_M\backslash\mathcal{M}^{(4+1)}_M$};
    \draw [-latex](3.7,0.3) -- (2.7,0.3);
    \draw (0,1) .. controls (0.8,0.5) and (2.2,0.5) .. (3,1);
    \draw[dashed] (0,1) .. controls (0.8,1.3) and (2.2,1.3) .. (3,1);
    \draw[very thick, shift={(1.65,0.78)}, rotate=30] (0,0) ellipse (0.08 and 0.24);
    \draw[very thick] (1.5,0.65) -- (1.5,0.92);
    \draw[very thick] (1.48,0.65) -- (1.7,0.58);
    \draw [-latex](0.6,-0.5) .. controls (0.8,-0.6) and (1,-0.65) .. (1.2,-0.68);
    \draw [-latex](0.6,-0.5) -- (0.6,0.1);
    \draw (1.1,-0.65) node[anchor=north] {$\upsilon$};
    \draw (0.6,0) node[anchor=east] {$t$};
\end{tikzpicture}
\end{center}

We notice that we can define $\underline{\upsilon}=\psi+t^*$, so $(t,z,\theta,\varphi,\underline{\upsilon})$ are also coordinates on $(t_+,\infty)\times(0,\infty)\times S^2\times S^1.$ In these coordinates, the metric is given by:
\[g^{(4+1)}_{M,a}=-2az^2dtd\underline{\upsilon}+dz^2+z^2t^2d\sigma_2^2+a^2z^2Wd\underline{\upsilon}^2.\]
Once again, we can extend as a vacuum Lorentzian metric to ${}^{(2)}\mathcal{M}^{(4+1)}_M=(t_+-\epsilon,\infty)\times(0,\infty)\times S^2\times S^1,$ and $\mathcal{CH}=\{t_+\}\times(0,\infty)\times S^2\times S^1$ represents a null Cauchy horizon of Taub--NUT type of $\big(\mathcal{M}^{(4+1)}_M,\ g^{(4+1)}_{M,a}\big)$.

The two extensions provided above are inequivalent. To prove this, we remark that the null curve $\{\upsilon(t,z_0,\theta_0,\varphi_0,\psi)=\upsilon_0\}\subset\mathcal{M}^{(4+1)}_M$ can be extended in ${}^{(1)}\mathcal{M}^{(4+1)}_M$ for $t<t_+$. However, in $(t,z,\theta,\varphi,\underline{\upsilon})$ coordinates on ${}^{(2)}\mathcal{M}^{(4+1)}_M,$ this curve is given by $\underline{\upsilon}=2t^*(t)+\upsilon_0$ and cannot be extended for $t<t_+$. We also point out that using the argument of \cite{HawkingEllis} we get that we cannot extend $\mathcal{M}^{(4+1)}_M$ to both ${}^{(1)}\mathcal{M}^{(4+1)}_M$ and ${}^{(2)}\mathcal{M}^{(4+1)}_M$ simultaneously, as the resulting manifold would not be Hausdorff. 

We summarize the above discussion into the following result:
\begin{proposition}\label{Taub-NUT extension}
    For any two parameters $M>-\frac{1}{3\sqrt{3}},\ a>0,$ the vacuum spacetime $\big(\mathcal{M}^{(4+1)}_M,\ g^{(4+1)}_{M,a}\big)$ has two inequivalent analytic extensions with closed timelike curves. The original spacetime has a null Cauchy horizon of Taub--NUT type $\mathcal{CH}$ which is reached along future directed causal curves in the limit $t\rightarrow t_++$.
\end{proposition}

For the rest of the section we prove that every solution of Theorem \ref{regular spacelike boundary} is given by one of the explicit solutions $\big(\mathcal{M}^{(4+1)}_M,\ g^{(4+1)}_{M,a}\big)$ in the cone exterior region. Using the Kaluza--Klein reduction from Section \ref{Kaluza Klein}, we obtain the corresponding Einstein-Scalar field solutions on $(t_+,\infty)\times(0,\infty)\times S^2$:
\[g_{M,a}=-\frac{az^3}{\sqrt{W}}dt^2+az\sqrt{W}dz^2+r^2d\sigma_2^2,\]
\[r_{M,a}=\sqrt{a}z^{3/2}W^{1/4}t,\]
\[\phi_{M,a}=\pm\frac{\sqrt{3}}{2}\log\big(az\sqrt{W}\big).\]
These are self-similar solutions, with $S=\frac{2}{3}z\partial_{z}$ and $k=\mp\frac{1}{\sqrt{3}}$. As previously remarked, the Einstein-Scalar field solutions have the gauge freedom $\phi\rightarrow\phi+C$. We notice that the corresponding vacuum solution is given by another spacetime in our family $\big(\mathcal{M}^{(4+1)}_M,\ g^{(4+1)}_{M,a'}\big)$.

We write our solutions in self-similar double null gauge:

\begin{proposition}
    In self-similar double null coordinates we have on $\big\{u<0,\ v\in(0,-u)\big\}\times S^2$:
    \[g_{M,a}=-\frac{4}{9}a\sqrt{W}dudv+r^2d\sigma_2^2,\]
    \[S=u\partial_u+v\partial_v,\]
    \[r_{M,a}=\sqrt{a}\sqrt{-uv}\cdot t W^{1/4}(u,v),\]
    \[\phi_{M,a}=\pm\frac{\sqrt{3}}{2}\log\big(a(-uv)^{\frac{1}{3}}\sqrt{W}\big).\]
\end{proposition}
\begin{proof}
We define the auxiliary null coordinates:
\[U(t,z)=\log z+\int_{t_+}^t\frac{d\tilde{t}}{\sqrt{W}},V(t,z)=\log z-\int_{t_+}^t\frac{d\tilde{t}}{\sqrt{W}}.\]
In $(U,V)$ coordinates we compute:
\[g_{M,a}=az^3\sqrt{W}dUdV+r^2d\sigma_2^2,\ S=\frac{2}{3}(\partial_U+\partial_V).\]
The desired null coordinates are $u=-e^{\frac{3}{2}U},\ v=e^{\frac{3}{2}V}.$
\end{proof}

We use the explicit form of the solutions to compute:
\[\beta=\frac{v\partial_vr}{r}=\frac{1}{3}\bigg(\frac{3}{2}-\frac{\sqrt{W}}{t}-\frac{W'}{4\sqrt{W}}\bigg)\]
\[\theta=r\frac{\partial_v\phi}{\partial_vr}=\frac{\pm\sqrt{3}t(W'-2\sqrt{W})}{4W+W't-6\sqrt{W}t}\]
\[1-\mu=9t^2\beta(1-\beta)\]
\[\eta=\beta\theta=\frac{1}{4\sqrt{3}}\frac{W'-2\sqrt{W}}{\sqrt{W}}\]

We want to view the above solutions as orbits of an autonomous system of ODEs. According to Section \ref{Bondi to double null coordinates}, the pair $(\theta,\beta)$ satisfies the system (\ref{double null systemmm}). We consider $e^{\underline{t}}=-v/u,$ which is equivalent to:
\[\underline{t}=-3\int_{t_+}^t\frac{1}{\sqrt{W}}.\]
With respect to the variable $\underline{t},$ the pair $(\theta,\beta)$ satisfies the system (\ref{System in t variablee}).

We notice that the coordinates $(t,z)$ cover the cone exterior region of the spacetime, i.e. the complement of the causal past of the scaling origin $\{u=0,\ v=0\}$. This follows since the cone exterior region is given by $\{u<0,\ v>0\}.$ The trapped region $\mathcal{T}^{(4+1)}$ is given by $\mu>1,$ which is equivalent to $\beta<0.$ Using the above formula for $\beta,$ we have $\mathcal{T}^{(4+1)}=\{t\in(t_+,t_{\mathcal{A}})\}$, where:
\[t_{\mathcal{A}}=\frac{1}{\sqrt{3}}+\frac{1}{\sqrt{3}}\sqrt{1+M3\sqrt{3}}.\]

We define as before:
\[\tau(\underline{t})=\begin{dcases}\tau_*+\int_{-\infty}^{\underline{t}}\beta d\underline{\tilde{t}}, &  \underline{t}\leq\underline{t}_{\mathcal{A}}\\
\tau_{\mathcal{A}}-\int_{\underline{t}_{\mathcal{A}}}^{\underline{t}}\beta d\underline{\tilde{t}}, &  \underline{t}>\underline{t}_{\mathcal{A}}
\end{dcases}\]
With respect to the variable $\tau,$ the pair $(\theta,\alpha)$ satisfies the system (\ref{system alpha theta}) in the region $\mathcal{E}^{(4+1)}\cap\mathcal{R}^{(4+1)}$ given by $\underline{t}\in(-\infty,\underline{t}_{\mathcal{A}})$, and it satisfies the system (\ref{system alpha theta trapped}) in the region $\mathcal{T}^{(4+1)}$ given by $\underline{t}\in(\underline{t}_{\mathcal{A}},0).$

We compute the limits:
\[r\xrightarrow[]{}0,\ \alpha\xrightarrow[]{}0-,\ \theta\xrightarrow[]{}-\sqrt{3}\text{ as }\tau\xrightarrow[]{}\infty,\]
\[\alpha\xrightarrow[]{}\infty,\ \theta\xrightarrow[]{}\infty,\ \eta\xrightarrow[]{}\eta_{\mathcal{A}}>0,\ \mu\xrightarrow[]{}1,\text{ as }\tau\xrightarrow[]{}\tau_{\mathcal{A}}-,\]
\[\alpha\xrightarrow[]{}\infty,\ \theta\xrightarrow[]{}\sqrt{3}\text{ as }\tau\xrightarrow[]{}\tau_*.\]
Thus, the orbits corresponding to the solutions $\big(\mathcal{M}_M^{(4+1)},g_{M,a}^{(4+1)}\big)$ have the same properties as the orbits corresponding to the solutions described in Theorem \ref{regular spacelike boundary} in the cone exterior region. We obtain the following result:

\begin{proposition}
    Consider any solution as in Theorem \ref{regular spacelike boundary}. Let $X\in\big(-\frac{1}{3},0\big)$ be the parameter from Lemma \ref{expansion at B lemma}. The orbit corresponding to this solution in the cone exterior region coincides with orbit corresponding to the explicit solution $\big(\mathcal{M}_M^{(4+1)},g_{M,a}^{(4+1)}\big)$ with:
    \[9t_+^2=-\frac{1}{X}.\]
\end{proposition}
\begin{proof}
For the explicit solutions defined above we have:
\[X=\lim_{t\xrightarrow[]{}t_+}\frac{\beta^2}{1-\mu}=-\frac{1}{9t_+^2}.\]
Moreover, for every $X\in\big(-\frac{1}{3},0\big)$ there exists a unique $M\in\big(-\frac{1}{3\sqrt{3}},\infty\big)$ such that $X=-\frac{1}{9t_+^2}.$
By the uniqueness result in Lemma \ref{expansion at B lemma}, we conclude that all the orbits of the solutions in Theorem \ref{regular spacelike boundary} are given by those of the explicit solutions $\big(\mathcal{M}_M^{(4+1)},g_{M,a}^{(4+1)}\big)$ in the cone exterior region.
\end{proof}

The above orbits determine the solution to the Einstein vacuum equations, once we specify the boundary values given by $r$ and $\phi$ on $C_0^-$. We compute in the case of Lorentzian Hawking--Page solutions:
\[r_m(u,s_*)=-u\sqrt{\frac{b}{2\pi}}\big(C_1(m)\big)^{\frac{3}{2}},\ \phi_m(u,s_*)=-k\log\bigg(-u\bigg(\frac{b}{2\pi}C_1(m)\bigg)^{\frac{3}{2}}\bigg).\]
We also compute in the case of the solutions $\big(\mathcal{M}_M^{(4+1)},g_{M,a}^{(4+1)}\big):$
\[r_{M,a}(u,\tau_*)=-u\sqrt{a}\big(C_2(M)\big)^{\frac{3}{2}},\ \phi_{M,a}(u,\tau_*)=-k\log\big(-u(aC_2(M))^{\frac{3}{2}}\big).\]

Using the gauge freedom (\ref{gauge freedom double null}), we can select $r_m=r_M,\ \phi_m=\phi_{M,a}$ on $C_0^-,$ so the solutions coincide in the cone exterior region. This allows us to complete the description of the solutions considered in Theorem \ref{regular spacelike boundary} at the boundary $\mathcal{B}:$ 

\begin{theorem}\label{regularity at BB}
    We consider a maximally extended Lorentzian Hawking--Page solution, which in the exterior regular region $\mathcal{E}^{(4+1)}\cap\mathcal{R}^{(4+1)}$ corresponds to an orbit to the right of $S_1^+$, the upper branch of the stable manifold of $P_1$ in the $(\theta,\alpha)$ plane. Let $X\in\big(-\frac{1}{3},0\big)$ be the parameter from Lemma \ref{expansion at B lemma}. Up to the gauge freedom (\ref{gauge freedom double null}), in the cone exterior region the solution coincides with the explicit solution $\big(\mathcal{M}_M^{(4+1)},g_{M,a}^{(4+1)}\big)$ with:
    \[9t_+^2=-\frac{1}{X}.\]
    Thus, the boundary $\mathcal{B}_{\text{Taub--NUT}}$ represents a null Cauchy horizon of Taub--NUT type, and beyond $\mathcal{B}_{\text{Taub--NUT}}$ the spacetime has at least two inequivalent analytic extensions with closed timelike curves.
\end{theorem}

Together with Theorems \ref{solutions singular null boundary}, \ref{solutions apparent horizon}, \ref{spacelike singularity} and \ref{regular spacelike boundary}, this result completes the classification of all the maximal extensions of Lorentzian Hawking--Page solutions in the class of $(4+1)$-dimensional scale-invariant globally hyperbolic vacuum spacetimes with an $SO(3)\times U(1)$ isometry, concluding the proof of the main result of our paper.

We recall that the $\big(\mathcal{M}_M^{(4+1)},g_{M,a}^{(4+1)}\big)$ solutions represent a Wick rotation of the Lorentzian Hawking--Page solutions. Thus, if we extend the Lorentzian Hawking--Page solutions in the cone exterior region by a Wick rotation according to the heuristics suggested in \cite{Andersoncone}, we obtain spacetimes with a null Cauchy horizon of Taub--NUT type. The generators of the Cauchy horizon are given by closed null curves, so the first singularity at the center $b_{\Gamma}$ is not a naked singularity. Moreover, we remark that none of the extensions of Lorentzian Hawking--Page solutions in our symmetry class lead to naked singularities. 

\begin{remark}
The spacetimes obtained in Theorem \ref{regularity at BB} are maximal globally hyperbolic future developments which can be extended as smooth Lorentzian manifolds beyond the null Cauchy horizon $\mathcal{B}_{\text{Taub--NUT}}$. In view of strong cosmic censorship, we conjecture this horizon to be unstable.
\end{remark}

\section{Solutions with Locally Naked Singularities}\label{third big section}

In this section we use the extensions of the Lorentzian Hawking--Page solutions to spacetimes with a null curvature singularity in order to construct asymptotically flat spacetimes with locally naked singularities. We argue that their instability in the $SO(3)\times U(1)$ symmetry class follows using the blue-shift effect of \cite{instability} and we address this in the context of Conjecture \ref{conjecture}.

\subsection{Asymptotically Flat Spacetimes}\label{asymptotically flat spacetimes}

We denote by $(\mathcal{M}_0^{(4+1)},g_0^{(4+1)})$ an extension of a Lorentzian Hawking--Page solution with a null curvature singularity as in Theorem \ref{solutions singular null boundary}. We recall that we have the decomposition $\mathcal{M}_0^{(4+1)}=\Gamma_0^{(4+1)}\cup\mathcal{U}_0^{(4+1)}\cup\mathcal{E}_0^{(4+1)},$ and the two dimensional Lorentzian quotient manifold  $\mathcal{Q}_0^{(1+1)}=\big(\mathcal{U}_0^{(4+1)}\cup\mathcal{E}_0^{(4+1)}\big)/\big(SO(3)\times U(1)\big).$ We also recall that $\big(\mathcal{Q}_0^{(1+1)}\times S^2,g_0,\phi_0\big)$ is the corresponding solution of the Einstein-Scalar field equations obtained via the Kaluza--Klein reduction. We notice that due to our scale invariance assumption, the spacetime $(\mathcal{M}_0^{(4+1)},g_0^{(4+1)})$ is not asymptotically flat. 

We consider an outgoing null cone $C_0^+,$ and an incoming null cone $C_1^-$ in the cone exterior region, so $C_1^-$ is to the future of $C_0^-.$ Using the Kaluza--Klein reduction, it suffices to prescribe characteristic initial data $(g,\phi)$ for the Einstein-Scalar field equations on $C_0^+.$ To the past of $C_1^-$ we consider the data induced by $\big(\mathcal{Q}_0^{(1+1)}\times S^2,g_0,\phi_0\big)$. We then use a cutoff to prescribe asymptotically flat data to the future of $C_1^-$, without a loss in regularity across the cone $C_1^-$. This can be done in double null coordinates by prescribing $(r,\phi)$ on $C_0^+$ to respect asymptotic flatness, and then solving the constraint equation (\ref{partial v}) with initial conditions on $C_0^+\cap C_1^-$ to obtain $\Omega.$ As a result, we obtained asymptotically flat characteristic initial data on $C_0^+$ for the $(4+1)$-dimensional Einstein vacuum equations.

We denote by $(\mathcal{M}^{(4+1)},g^{(4+1)})$ the evolution of the above initial data. Since we solved the reduced system, the spacetime obtained has an $SO(3)\times U(1)$ isometry. As above, we consider the quotient manifold $\mathcal{Q}^{(1+1)},$ and the corresponding solution of the Einstein-Scalar field equations $\big(\mathcal{Q}^{(1+1)}\times S^2,g,\phi\big)$. By domain of dependence, the spacetime $(\mathcal{M}^{(4+1)},g^{(4+1)})$ coincides with the original Lorentzian Hawking--Page extended solution in the past of $C_1^-$. Therefore, the new spacetime has a portion of the null boundary $\mathcal{B}_{null}$ where the curvature is singular. This singularity is not preceded by trapped surfaces, because in the case of the corresponding Einstein-Scalar field solution $\big(\mathcal{Q}^{(1+1)}\times S^2,g,\phi\big)$ the null boundary is also not preceded by trapped surfaces. Thus, the spacetime $(\mathcal{M}^{(4+1)},g^{(4+1)})$ has a locally naked singularity. We represent the Penrose diagram of the quotient
of the spacetime by the $SO(3)\times U(1)$ action below.

The Einstein-Scalar field solution that we constructed satisfies the usual notion of asymptotic flatness. Using the Kaluza--Klein reduction, we get that the vacuum spacetime $(\mathcal{M}^{(4+1)},g^{(4+1)})$ approaches the standard flat Lorentzian metric on $\mathbb{R}^4\times S^1$ at infinity. In conclusion, we constructed asymptotically flat extensions of the Lorentzian Hawking--Page solutions with locally naked singularities and an $SO(3)\times U(1)$ isometry.

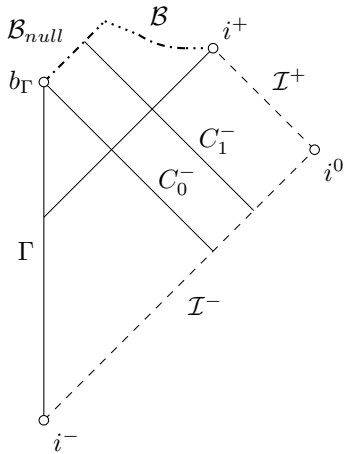
\begin{center}
\begin{tikzpicture}[scale=0.9]
\draw[dashed] (0,0) -- (4,4);
\draw (0,0) -- (0,5);
\draw[very thin] (0,3) -- (2.5,5.5);
\draw[very thin] (0,5) -- (2.5,2.5);
\draw[dashed] (4,4) -- (2.5,5.5);
\draw[thick, dashdotted] (0,5) -- (0.85,5.85);
\draw[thick, dotted] (0.85,5.85) -- (0.9,5.9);
\draw[very thin] (0.6,5.6) -- (3.1,3.1);
\draw[thick, dashdotted] (1.4,5.7) .. controls (1.7,5.5) and (2,5.5) .. (2.1,5.5);
\draw[thick, dotted] (2.1,5.5) -- (2.5,5.5);
\draw[thick, dotted] (0.9,5.9) -- (1.4,5.7);
\filldraw[color=black, fill=white] (0,5) circle (2pt) node[anchor=east] {$b_{\Gamma}$};
\filldraw[color=black, fill=white] (2.5,5.5) circle (2pt) node[anchor=south west] {$i^+$};
\filldraw[color=black, fill=white] (4,4) circle (2pt) node[anchor=north west] {$i^0$};
\filldraw[color=black, fill=white] (0,0) circle (2pt) node[anchor=north west] {$i^-$};
\draw (0,2.5)  node[anchor=east] {$\Gamma$};
\draw (1.55,3.55)  node[anchor=west] {$C_0^-$};
\draw (2.15,4.15)  node[anchor=west] {$C_1^-$};
\draw (2,2)  node[anchor=north west] {$\mathcal{I}^-$};
\draw (3.25,4.75)  node[anchor=south west] {$\mathcal{I}^+$};
\draw (0.45,5.45)  node[anchor=south east] {$\mathcal{B}_{null}$};
\draw (1.7,5.7)  node[anchor=south] {$\mathcal{B}$};
\end{tikzpicture}    
\end{center}

\begin{remark}
    In the above Penrose diagram we represented the case when the spacelike or Taub--NUT-like component of the boundary $\mathcal{B}$ is nonempty. However, we cannot rule out the situation when $\mathcal{B}=\emptyset$ and the null boundary $\mathcal{B}_{null}$ extends all the way to $i^+.$ In this case, we use the monotonicity properties of the Einstein-Scalar field equations in spherical symmetry to get that $r$ extends continuously to zero along $\mathcal{B}_{null}$. We claim that in both cases we can use the argument of \cite{trappedsurface} to prove that the corresponding Einstein-Scalar field solution has a complete future null infinity $\mathcal{I}^+.$ Using also the decay of the scalar field, we can then use this argument to prove that in both cases the vacuum spacetime $(\mathcal{M}^{(4+1)},g^{(4+1)})$ has a complete future null infinity $\mathcal{I}^+.$
\end{remark}

\begin{remark}
    We could also repeat the above argument for the solutions constructed in Theorem \ref{spacelike singularity} and Theorem \ref{regular spacelike boundary}, to obtain the corresponding asymptotically flat extensions of the Lorentzian Hawking--Page solutions. In this case, the asymptotic behavior that we computed holds near a portion of $\mathcal{B}$, in the past of $C_1^-$. We can again use \cite{trappedsurface} as above to prove that these extensions have a complete future null infinity $\mathcal{I}^+.$ This argument also shows that we can modify the associated Einstein-Scalar field solutions of \cite{nakedsingularities} with $k^2=\frac{1}{3}$ to asymptotically flat solutions with spacelike singularities. For a portion of spacelike singularity near the center, the blow-up rate of the scalar field is $\phi=c\log r+O(1)$, where $c=\pm\sqrt{3},\ \mp\frac{1}{\sqrt{3}}.$
\end{remark}

\begin{remark}
    We consider a maximal extension of a Lorentzian Hawking–Page solution, as constructed in Section \ref{second big section}, and we fix an incoming null cone $C_1^-$ in the exterior region. As a consequence of the above two remarks, we get that any $(4+1)$-dimensional asymptotically flat vacuum solution with an $SO(3)\times U(1)$ isometry, that coincides with our extended spacetime in the past of $C_1^-$, has a complete future null infinity $\mathcal{I}^+.$
\end{remark}

\subsection{Instability in the $SO(3)\times U(1)$ Symmetry Class}

In this section we prove that the above constructed asymptotically flat spacetime $(\mathcal{M}^{(4+1)},g^{(4+1)})$ with a locally naked singularity is unstable in the $SO(3)\times U(1)$ symmetry class. We point out that showing the initial data is non-generic in a whole well-posedness class is beyond the purpose of this paper. On the other hand, we illustrate how the blue-shift instability mechanism of \cite{instability} applies in our case, by providing a two dimensional set of perturbations of our solution to spacetimes without locally naked singularities.

Given a spherically symmetric solution of the Einstein-Scalar field equations $(M,g,\phi)$ with a first singularity at the center $b_{\Gamma},$ we denote by $C_0^-$ the incoming null cone going into $b_{\Gamma}$ and by $C_0^+$ an outgoing null cone. We write the metric in double null coordinates $(u,v)$ such that $C_0^-=\{v=0\},$ $u=-2r$ on $C_0^-,$ and $v=2r$ on $C_0^+.$ Note that with this normalization we have $b_{\Gamma}=\{u=0,\ v=0\}.$ The blue-shift is defined:
\[\gamma(u)=\int_{u_0}^u\frac{1}{r}\cdot\frac{\mu}{1-\mu}(u',0)du'.\]

As pointed out before, given our normalization it suffices to prescribe $\phi$ on $C_0^+$ to obtain characteristic initial data. Using the two instability theorems of \cite{instability}, one gets the following result:
\begin{theorem}[D. Christodoulou, \cite{instability}]
    Let $(M,g,\phi)$ be a bounded variation spherically symmetric solution of the Einstein-Scalar field equations. Assume that the spacetime has a locally naked singularity and that the blue-shift is unbounded on $C_0^-$. There exist functions $f_1,\ f_2$ on $C_0^+$ with $\partial_vf_1\in BV,\ f_2\in C^{1,1/2}$ and $f_1=f_2=0$ for $v<0$, with the following property: for any two constants $\lambda_1,\ \lambda_2$ such that at least one is nonzero, any exterior incoming null cone in the future development of the initial data $\phi+\lambda_1f_1+\lambda_2f_2$ reaches a spacelike singularity $\mathcal{B}.$
\end{theorem}

The instability theorem gives a 2-dimensional space of perturbations of solutions with locally naked singularities and unbounded blue-shift to spacetimes without locally naked singularities, where the boundary $\mathcal{B}$ is preceded by trapped surfaces. We represent the Penrose diagram of such Einstein-Scalar field solutions:
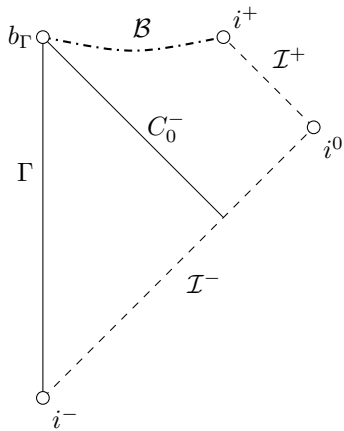
\begin{center}
\begin{tikzpicture}[scale=1.2]
\draw[dashed] (0,0) -- (3,3);
\draw[dashed] (3,3) -- (2,4);
\draw (0,0) -- (0,4);
\draw[very thin] (0,4) -- (2,2);
\draw[thick, dashdotted] (0,4) .. controls (1,3.8)  .. (2,4);
\filldraw[color=black, fill=white] (0,4) circle (2pt) node[anchor=east] {$b_{\Gamma}$};
\draw (0,2.5)  node[anchor=east] {$\Gamma$};
\draw (1.05,3)  node[anchor=west] {$C_0^-$};
\draw (1.5,1.5)  node[anchor=north west] {$\mathcal{I}^-$};
\draw (1.3,3.9)  node[anchor=south east] {$\mathcal{B}$};
\filldraw[color=black, fill=white] (3,3) circle (2pt) node[anchor=north west] {$i^0$};
\filldraw[color=black, fill=white] (2,4) circle (2pt) node[anchor=south west] {$i^+$};
\filldraw[color=black, fill=white] (0,0) circle (2pt) node[anchor=north west] {$i^-$};
\draw (2.45,3.55)  node[anchor=south west] {$\mathcal{I}^+$};
\end{tikzpicture}
\end{center}

\begin{corollary}
    Given $(\mathcal{M}^{(4+1)},g^{(4+1)})$ our extension of a Lorentzian Hawking--Page solution to an asymptotically flat spacetime with a locally naked singularity, there exists in the $SO(3)\times U(1)$ symmetry class a two dimensional space of perturbations to globally hyperbolic spacetimes without locally naked singularities, where the boundary $\mathcal{B}$ is preceded by trapped surfaces. 
\end{corollary}
\begin{proof}
    We consider the Einstein-Scalar field solution $\big(\mathcal{Q}^{(1+1)}\times S^2,g,\phi\big)$ associated to the $(4+1)$-dimensional vacuum spacetime $(\mathcal{M}^{(4+1)},g^{(4+1)})$. This solution has a locally naked singularity, and $\mu=\frac{1}{4}$ on $C_0^-,$ so the blue-shift is unbounded. However, the center $\Gamma^{(3+1)}$ is singular, so we cannot apply the theorem directly. We consider the characteristic initial data $\phi+\lambda_1f_1+\lambda_2f_2$ on $C_0^+,$ and obtain a unique future development of the data with the same argument that we used in Section \ref{Extending Beyond the Interior Region section}. We then notice that the proof of the above theorem depends only on the solution on $C_0^-$ and the cone exterior region, so we obtain that any exterior incoming null cone in the perturbed solution reaches the trapped region. Using the Kaluza--Klein reduction, we have the same behavior for the corresponding $(4+1)$-dimensional vacuum solution.
\end{proof}

\begin{remark}
    In \cite{instability}, Christodoulou proves Conjecture \ref{conjecture} for the Einstein-Scalar field system in spherical symmetry, which then implies weak cosmic censorship for this matter model. In addition to the instability theorem, Christodoulou also proves that for AC solutions with locally naked singularities the blue-shift is unbounded, and that the one-dimensional spaces of perturbations obtained using $f_2$ are pairwise disjoint. In the case of the $(4+1)$-dimensional vacuum equations with an $SO(3)\times U(1)$ symmetry, one could also prove weak cosmic censorship as a consequence of Conjecture \ref{conjecture}. In this context, our above proof illustrates how the blue-shift effect implies the instability of the locally naked singularities that we constructed. However, to prove the conjecture we also need an extension principle at the center to conclude that given a first singularity at the center the blue-shift is unbounded.
\end{remark}

\begin{remark}
We could also consider proving the instability of the solutions with locally naked singularities $(\mathcal{M}^{(4+1)},g^{(4+1)})$ outside of symmetry. A potential candidate for perturbing could be constructed starting from localizing the functions $f_1,\ f_2$ near a fixed angle $(\theta_0,\varphi_0,\psi_0).$ An essential step for this problem is a formation of trapped surfaces result analogue to \cite{formationoftrappedsurfaces}, but without symmetry assumptions. We remark that the short-pulse solutions of \cite{shortpulse} cannot be used, since we need to allow for small perturbations and use the unbounded blue-shift on $C_0^-.$ Since the incoming cone is far from Minkowski space, the situation is closer to that of \cite{impulsive}. Finally, we point out that the problem of perturbing spherically symmetric naked singularities of the Einstein-Scalar field system outside of symmetry was considered in \cite{Li-Liu}.
\end{remark}

\section{Appendix}

\subsection{Expansion of the solutions near the center}\label{Appendix Expansion Subsection}

In this section we compute in detail the expansions used in Section \ref{expansion near center section}. The goal is to study the autonomous system of ODEs:
\begin{equation}\label{appendix system}
    \begin{dcases}
    \frac{d\alpha}{ds}=\alpha\big[(\theta+k)^2+(1-k^2)(1-\alpha)\big] \\
    \frac{d\theta}{ds}=k\alpha\big(k\theta-1\big)+\theta\big[(\theta+k)^2-(1+k^2)\big]
    \end{dcases}
\end{equation}
in a neighborhood of the critical point $(\alpha,\theta)=(0,\pm\sqrt{3}).$ We set $x_1=\alpha,\ x_2=\theta\mp\sqrt{3},$ to obtain the system:
\begin{equation}
    \begin{dcases}
    \frac{dx_1}{ds}=2x_1+x_1\bigg[x_2^2\pm\frac{4}{\sqrt{3}}x_2-\frac{2}{3}x_1\bigg] \\
    \frac{dx_2}{ds}=-2kx_1+4x_2+\frac{1}{3}x_1x_2\pm\frac{7}{\sqrt{3}}x_2^2+x_2^3
    \end{dcases}
\end{equation}
We diagonalize the linear part of this system by setting $y_1=x_1,\ y_2=-kx_1+x_2$. We obtain the system:
\begin{equation}\label{diagonal system}
    \begin{dcases}
    \frac{dy_1}{ds}=2y_1-2y_1^2\pm\frac{4}{\sqrt{3}}y_1y_2+U_1^{(3)}(y_1,y_2) \\
    \frac{dy_2}{ds}=4y_2\pm\frac{7}{\sqrt{3}}y_2^2-3y_1y_2+U_2^{(3)}(y_1,y_2)
    \end{dcases}
\end{equation}
where $U_1^{(3)}(y_1,y_2),\ U_2^{(3)}(y_1,y_2)$ are homogeneous polynomials of degree 3. The critical point $(0,0)$ is stable to the past. The eigenvalues $\lambda_1=2,\ \lambda_2=4$ are resonant, since $(2,0)\cdot(\lambda_1,\lambda_2)=\lambda_2.$ Using the Poincaré-Dulac theorem, we obtain that the solutions have the form:
\begin{equation}
    \begin{dcases}
    y_1=z_1+Q_1^{(2)}(z_1,z_2) \\
    y_2=z_2+Q_2^{(2)}(z_1,z_2)
    \end{dcases}
\end{equation}
where $Q_1^{(2)}(z_1,z_2),\ Q_2^{(2)}(z_1,z_2)$ are analytic function that vanish to order 2 at the origin, and $z_1,\ z_2$ are solutions to the system:
\begin{equation}
    \begin{dcases}
    \frac{dz_1}{ds}=2z_1\\
    \frac{dz_2}{ds}=4z_2+C_0z_1^2
    \end{dcases}
\end{equation}
We compute $z_1=a_1e^{2s},\ z_2=a_2e^{4s}+C_0a_1^2se^{4s}.$ We introduce the equivalence relation $(a_1,a_2)\sim(Aa_1,A^2a_2)$, for some $A>0,$ and we denote by $\mathfrak{a}$ the set of equivalence classes. There is a one to one correspondence between $\mathfrak{a}$ and the orbits of (\ref{diagonal system}) near the origin.  We write:
\[Q_1^{(2)}(z_1,z_2)=b^{ij}z_iz_j+Q_1^{(3)}(z_1,z_2),\ Q_2^{(2)}(z_1,z_2)=d^{ij}z_iz_j+Q_1^{(3)}(z_1,z_2),\] and we obtain that $C_0=0$. This corresponds to the fact that even though the eigenvalues are resonant, the resonant monomial $\begin{pmatrix} 0 & y_1^2 \end{pmatrix}$ does not appear on the right hand side of (\ref{diagonal system}). We also obtain $b^{11}=-1.$ Thus:
\[y_1=a_1e^{2s}-a_1^2e^{4s}+O(e^{6s}),\ y_2=\big(a_2+a_1^2d^{11}\big)e^{4s}+O(e^{6s}).\]
Finally, in a neighborhood of the negatively stable critical point $(0,\pm\sqrt{3})$, the orbits of (\ref{appendix system}) satisfy the expansion:
\[\alpha=a_1e^{2s}-a_1^2e^{4s}+O(e^{6s}),\]
\[\theta=\pm\sqrt{3}\mp\frac{a_1}{\sqrt{3}}e^{2s}+\bigg(\pm\frac{a_1^2}{\sqrt{3}}+a_1^2d^{11}+a_2\bigg)e^{4s}+O(e^{6s}),\]
for some $(a_1,a_2)\in\mathfrak{a}.$

\subsection{Sectional curvature computation}\label{Appendix sectional curvature}
In this section, we compute the sectional curvatures of $S^2$ and $\mathcal{Q}^{(1+1)}$ with respect to the metric:
\[\Tilde{g}=-\frac{1}{f}e^{2\nu}du^2-\frac{2}{f}e^{\nu+\lambda}dudr+\frac{1}{f}r^2d\sigma_2^2+f^2d\sigma_1^2.\]
We express these using the corresponding $4$-dimensional metric:
\[g=-e^{2\nu}du^2-2e^{\nu+\lambda}dudr+r^2d\sigma_2^2,\]
and the quantities defined in Section \ref{introduction to bondi gauge}. 

In general, we compute for $a,b,c,d\neq\psi:$
\[\tilde{R}_{abc}{ }^d=R_{abc}{ }^d-\delta^d{ }_{[a}\nabla_{b]}\nabla_c\log f+g_{c[a}\nabla_{b]}\nabla^d\log f+\frac{1}{2}\big(\nabla_{[a}\log f\big)\delta^d{ }_{b]}\nabla_c\log f-\]\[-\frac{1}{2}\big(\nabla_{[a}\log f\big)g_{b]c}\nabla^d\log f-\frac{1}{2}g_{c[a}\delta^d{ }_{b]}g(\nabla\log f,\nabla\log f).\]

We apply this formula to compute:
\[\tilde{K}_{S^2}=\frac{f}{r^2}\tilde{R}_{\theta\varphi\theta}{ }^{\varphi}=\frac{f}{r^2}\bigg[R_{\theta\varphi\theta}{ }^{\varphi}+\nabla_{\theta}\nabla_{\theta}\log f-\frac{1}{4}r^2g(\nabla\log f,\nabla\log f)\bigg]=\]\[=\frac{f}{r^2}\bigg[\frac{2m}{r}-k(1-\mu)(\theta+\zeta)-k^2r^2g(\nabla\phi,\nabla\phi)\bigg]=\frac{f}{r^2}\bigg[1-(1-\mu)(k\theta+1)(k\zeta+1)\bigg].\]

Similarly, we also compute:
\[\tilde{K}_{\mathcal{Q}^{(1+1)}}=\frac{f^2}{4}\tilde{R}(n,l,n,l)=f^2K_{\mathcal{Q}^{(1+1)}}+\frac{f^2}{2}\nabla_{l}\nabla_{n}\log f+\frac{f^2}{4}\nabla_{l}\log f\nabla_{n}\log f+\frac{f^2}{4}g(\nabla\log f,\nabla\log f)=\]\[=\frac{f}{r^2}\bigg[1-(1-\mu)(k\theta+1)(k\zeta+1)\bigg]+(1+k^2)\frac{f^2}{r^2}(1-\mu)\theta\zeta.\]

\subsection{Equations in double null coordinates}\label{Appendix Double Null}

In this section, we write the Einstein-Scalar field equations in spherical symmetry for a self-similar solution in double null coordinates as a system of autonomous ODEs. We consider the metric:
\[g=-\Omega^2dudv+r^2d\sigma_2^2.\]
The Einstein-Scalar field equations are:

\begin{equation}\label{partial u}
    \partial_u(\Omega^{-2}\partial_ur)=-r\Omega^{-2}(\partial_u\phi)^2
\end{equation}
\begin{equation}\label{partial v}
    \partial_v(\Omega^{-2}\partial_vr)=-r\Omega^{-2}(\partial_v\phi)^2
\end{equation}
\begin{equation}\label{partial uv Omega}
    \partial_u\partial_v\log\Omega=-\frac{1}{r}\partial_u\partial_v r-4\partial_u\phi\partial_v\phi
\end{equation}
\begin{equation}\label{partial uv r}
    \partial_u\partial_vr^2=-\frac{1}{2}\Omega^2
\end{equation}
\begin{equation}\label{wave}
    r\partial_u\partial_v\phi+\partial_v\phi\partial_ur+\partial_u\phi\partial_vr=0
\end{equation}
We notice that equation (\ref{partial uv Omega}) is implied by the other equations. Moreover, we can rewrite the Raychaudhuri's equations (\ref{partial u}), (\ref{partial v}) as:
\begin{equation}\label{partial u m}
    \partial_um=-2r^2\Omega^{-2}\partial_vr(\partial_u\phi)^2
\end{equation}
\begin{equation}\label{partial v m}
    \partial_vm=-2r^2\Omega^{-2}\partial_ur(\partial_v\phi)^2
\end{equation}
where $m$ is the Hawking mass:
\begin{equation}\label{mass}
    m=\frac{r}{2}(1+4\Omega^{-2}\partial_ur\partial_vr).
\end{equation}

Following Section \ref{Bondi to double null coordinates}, we assume we are in self-similar double null gauge, so $S=u\partial_u+v\partial_v$ and:
\[S\Omega=0,\ Sr=r,\ S\phi=-k.\]
As before, we consider the following quantities:
\[\theta=\frac{r\partial_v\phi}{\partial_vr},\ \zeta=\frac{r\partial_u\phi}{\partial_ur},\ \beta=\frac{v\partial_vr}{r}.\]
We introduce the self-similar coordinate:
\[y=-\frac{v}{u}.\]
The change of coordinates is given by:
\[\partial_v=-\frac{1}{u}\partial_y,\ \partial_u=\frac{v}{u^2}\partial_y+\partial_u.\]
Thus, in $(u,y)$ coordinates we have $S=u\partial_u.$ We recall the notation:
\[\chi=\phi+k\log(-u),\ R=-\frac{r}{u}.\]
The self-similar condition is equivalent to $\Omega=\Omega(y),\ R=R(y),\ \chi=\chi(y).$ We compute:
\[\beta=\frac{y\partial_yR}{R},\ \zeta=\frac{\theta\beta+k}{\beta-1},\ \Omega^2=-\frac{4}{1-\mu}\partial_yR\big(y\partial_yR-R\big).\]
We can rewrite the wave equation (\ref{wave}) as:
\[y\partial_y\theta+\theta\bigg(1+\frac{\beta-1}{1-\mu}\bigg)+k=0.\]
Also, we can rewrite equation (\ref{partial uv r}) as:
\[y\partial_y\beta+\bigg(2-\frac{1}{1-\mu}\bigg)\beta(\beta-1)=0.\]
Using this in equations (\ref{partial u m}),(\ref{partial v m}) eventually gives:
\[\frac{1}{1-\mu}=\beta\theta^2-\frac{(\theta\beta+k)^2}{\beta-1}+1.\]
Finally, we conclude that the Einstein-Scalar field system is equivalent to:
\[\begin{dcases}
    y\partial_y\theta=k\big(k\theta-1\big)+\beta\theta\big[(\theta+k)^2-(1+k^2)\big] \\
    y\partial_y\beta=\beta(1-k^2)-\beta^2\big[(\theta+k)^2+(1-k^2)\big]
    \end{dcases}\]

\bibliographystyle{alpha}
\bibliography{refs}

\end{document}